\documentclass[aps,pra,twocolumn,a4paper,nofootinbib,floatfix]{revtex4-2}
\usepackage[utf8]{inputenc}
\usepackage[english]{babel}
\usepackage{amsmath}
\usepackage{mathtools}
\usepackage{amsfonts}
\usepackage{amssymb}
\usepackage{amsthm}
\usepackage{graphicx}
\usepackage{enumerate}
\usepackage{hyperref}
\usepackage{braket}
\hypersetup{citecolor=blue}
\hypersetup{colorlinks=true}
\hypersetup{linkcolor=blue}
\hypersetup{urlcolor=blue}
\usepackage{caption}
\usepackage{subcaption}
\newcommand*{\Lie}{\mathcal{L}}
\newcommand*{\NN}{\mathbb{N}}
\newcommand*{\QT}{\mathrm{QT}}
\newcommand*{\RR}{\mathbb{R}}
\newcommand*{\ZZ}{\mathbb{Z}}

\DeclareMathOperator{\e}{e}
\DeclareMathOperator{\supp}{supp}
\DeclarePairedDelimiter\abs\lvert\rvert

\DeclarePairedDelimiter\mb{\lbrace\!\lbrace}{\rbrace\!\rbrace}
\newcommand*{\dd}{\mathrm{d}}
\newcommand*{\ddd}[1]{\mathrm{d}^{#1} \!}
\newcommand*{\gnd}{\mathrm{gnd}}
\newcommand*{\eff}{\mathrm{eff}}
\DeclareMathOperator{\ini}{in}
\newtheorem{theorem}{Theorem}
\newtheorem{corollary}[theorem]{Corollary}
\newtheorem{example}[theorem]{Example}
\newtheorem{proposition}[theorem]{Proposition}

\begin{document}
\title{Generalized dynamical theories in phase space and the hydrogen atom}

\author{Martin Pl\'{a}vala}
\email{martin.plavala@uni-siegen.de}
\affiliation{Naturwissenschaftlich-Technische Fakultät, Universität Siegen, 57068 Siegen, Germany}

\author{Matthias Kleinmann}
\affiliation{Naturwissenschaftlich-Technische Fakultät, Universität Siegen, 57068 Siegen, Germany}

\begin{abstract}
We show that the phase-space formulation of general probabilistic theories can be extended to include a generalized time-evolution and that it can describe a nonquantum hydrogen-like system which is stable, has discrete energy levels, and includes the Zeeman effect. This allows us to study dynamical effects such as excitations of the hydrogen-like system by a resonant laser and Rutherford scattering. Our construction demonstrates that classical theory and quantum theory can be seen as specific choices of general probabilistic theory in phase space and that other probabilistic theories also lead to measurable predictions.
\end{abstract}

\maketitle

\section{Introduction}
General probabilistic theories (GPTs) are a widely accepted mathematical 
framework that includes classical and quantum theories as special cases: this 
is the case in operational settings 
\cite{JanottaHinrichsen-review,Lami-thesis,Mueller-review,Plavala-review,Leppajarvi-thesis} 
and recently it was discovered that the same is true for harmonic oscillator 
\cite{PlavalaKleinmann-oscillator}. GPTs allow us to compare the computational 
power of classical and quantum theories 
\cite{BarrettBeaudrapHobanLee-GPTcomputation,Perinotti-highOrderComputation,GilliganLee-computation} 
and to identify key non-classical structures that enable a quantum advantage 
\cite{Spekkens-contextuality,MatsumotoKimura-superInformationStorability,SelbySikora-money,HeinosaariLeppajarviPlavala-noFreeInformation,SikoraSelby-coinFlipping,SchmidSpekkens-stateDiscrimination,DArianoErbaPerinotti-classicalEntanglement}, 
but also to develop experimental tests of the structure of quantum theory 
\cite{Mazurek_2021,RenouTrilloWeilenmannLeTavakoliGisinAcinNavascues-realQT,MuellerGarner-testingQT}. 
The framework of GPTs also includes exotic theories that were previously used 
for analysis of violations of Bell inequalities \cite{Barrett-GPTinformation}, 
analysis of steering \cite{CavalcantiSelbySikoraGalleySainz-witworld}, as well 
as to formulate device-independent cryptographic protocols 
\cite{Zhang_2022,Nadlinger_2022}. One can also add additional postulates in 
order to derive finite-dimensional quantum theory 
\cite{Hardy-derivationQT,MasanesMuller-derivationQT,ChiribellaDArianoPerinotti-derivationQT,Wilce-derivationQT,Buffenoir-derivationQT}.

However, one faces two main difficulties when one aims to describe physical systems like the hydrogen atom within a GPT. First, one needs to construct a model in the formalism of GPTs for which one can reasonably argue that it constitutes a model of the hydrogen atom. The problem of model building is addressed by a phase-space formulation generalizing the Wigner-Weyl formalism in quantum theory \cite{Spekkens-epistricted,LinDahlsten-tunneling,PlavalaKleinmann-oscillator}. Second, in order to obtain predictions beyond simple static predictions, one needs to formulate the time-evolution of the system. In quantum theory, a general Markovian time-evolution is given by the Lindblad equation. In fact, besides Markovianity the Lindblad equation is a consequence of the assertion that the system is described by a quantum state, for all times. While Markovianity can be directly formulated in GPTs \cite{GrossMuellerColbeckDahlsten-dynamicsGPT,AlSafiShort-boxworldDynamics,AlSafiRichens-reversibleDynamics,BranfordDahlstenGarner-hamiltonDynamics}, the positivity conditions are more difficult to implement. This difficulty has at least two sources. First, the set of states may not be fully specified, for example, when no exhaustive set of observables has been described. Second, even if the set of states is known, positivity alone may not be a sufficient condition for similar reasons that in quantum theory one needs to consider complete positivity \cite{HeinosaariZiman-MLQT}. We avoid this difficulties by imposing a general form of the dynamical equations and subordinating positivity to the time evolution. The time-evolution is based on the Moyal bracket, but has a significantly more general form. Despite of its generality, we show that any GPT under this time-evolution exhibits a generalized version of Ehrenfest's theorem.

Returning to the hydrogen atom, we first provide a spectral decomposition of the energy observable, study the ground state, determine the conserved quantities, and discuss the degeneracy of the spectrum. Then we consider two dynamical situations, first the interaction of the electron with an external electric field, and second the scattering of a charged particle in a Rutherford-type scenario. In the first situation, there are deviations to the quantum predictions, albeit quite small. For the scattering, we use a Green's function approach. We find that, in the far field, the scattering cannot be different from the classical and quantum case, thus yielding the Rutherford scattering in any generalized dynamical theory in phase-space.

\section{Operational theories in phase space}
We build on the results of Ref.~\cite{PlavalaKleinmann-oscillator}. In this framework, the state of a physical system is described by a real-valued phase-space function $\rho(\vec{q}, \vec{p})$. This function generalizes ensembles in statistical mechanics by allowing that $\rho$ can attain negative values for some regions in phase space. It follows that $\rho$ does no longer have an interpretation as a probability density over phase space. But, in analogy to Wigner functions in quantum theory, the marginals $\rho$ are required to have an interpretation as a probability density. That is,
\begin{equation}
\rho_q(\vec{q}) = \int_{\RR^3} \rho(\vec{q},\vec{p})\ddd{3} p,
\end{equation}
is a proper probability density for position obeying $\rho_q(\vec{q}) \geq 0$ for all $\vec{q}$ and $\int \rho_q(\vec{q})\ddd{3} q = 1$ and similarly for the momentum marginal $\rho_p(\vec{p})$. Note that we always consider a phase-space of dimension $3+3$, with generalizations to other dimensions being straightforward.

In addition to the marginal distributions, the state of a system determines the expectation value of any observable. An observable itself is described by the phase-space function known from classical mechanics, for example, $H(\vec{q},\vec{p})=\frac{\abs{\vec{p}}^2}{2 \mu}$ for the energy of a free nonrelativistic particle or $L_3(\vec{q},\vec{p})=q_1 p_2 - q_2 p_1$ for the $z$-component of the angular momentum vector. In order to be more concrete about the ``expectation value of an observable,'' we describe the outcome of a measurement of an observable $A$ as a random variable $\tilde{A}$. The expectation value of $\tilde{A}$ is then computed via
\begin{equation}
\braket{\tilde{A}}= \int_{\RR^6} A(\vec q,\vec{p})\rho(\vec q,\vec{p})\ddd{3} q \ddd{3} p= \braket{A,\rho},
\end{equation}
where we introduced the short-hand notation $\braket{A,\rho}$.

We mention that although the phase-space observables remain the same as in 
classical theory, their interpretation does change: In classical theory, 
$A(\vec{q},\vec{p})$ can be understood as the value of $A$ if the system is at 
the phase-space point $(\vec{q},\vec{p})$. This interpretation does not hold in 
the generalized framework, since $\rho$ is not a probability density over the 
phase space. Rather $A(\vec{q},\vec{p})$ should be seen as an abstract 
representation of an observable as a function over the phase space. As a 
consequence, the second moment of the random variable $\tilde{A}$, that is, the 
mean square, cannot be computed as an integral over $A^2(\vec{q},\vec{p})$. In 
other words in general we have $\braket{A^2, \rho} \neq 
\braket{(\tilde{A})^2}$.

We can deduce already a first implication of the phase-space formalism: It is a 
known fact stated in textbooks on quantum theory 
\cite{Cohen-TannoudjiDiuLaloe-quantum,Bohm-quantum,GottfriedYan-quantum} that 
the electron in the hydrogen atom does not collapse into the nucleus, because 
that would violate uncertainty relations. Hence the ground state energy of the 
quantum hydrogen atom must be finite. Here we derive the other implication, 
showing that uncertainty relations between position and momentum are necessary 
to prevent the hydrogen atom from collapsing in operational theories on phase 
space. For this, we compare the minimal value of the energy at any phase-space 
point $E_\mathrm{min} = \inf_{\vec{q},\vec{p}} H(\vec{q}, \vec{p})$ with the 
lowest energy reachable among all states of a given theory, $E_0 = \inf_\rho 
\braket{H,\rho}$. We say that a theory exhibits preparation uncertainty, if the 
set of states is restricted in such a way that $\rho(\vec{q},\vec{p})= 
\delta^{(3)}(\vec{q}_0-\vec{q})\delta^{(3)}(\vec{p}_0-\vec{p})$ is not a valid 
state for some phase-space point $(\vec{q}_0, \vec{p}_0)$. This kind of 
uncertainty is clearly a necessary precondition for $E_0 > E_\mathrm{min}$, 
thus proving our claim.

Finally, we are interested in the probability distribution of the random variable $\tilde{A}$. For this one invokes a ``spectral decomposition'' of $A(\vec{q},\vec{p})$ in the form of the phase-space spectral measure $g_A(I;\vec{q},\vec{p})$. The spectral measure yields then the probability distribution of $\tilde{A}$ via
\begin{equation}
\Pr[\tilde{A}\in I] = \braket{g_A(I),\rho}.
\end{equation}
Here, $I$ is any measurable subset of the range of $\tilde{A}$. The phase space spectral measure must satisfy two properties: $\Pr[\tilde{A}\in I]$ must be a probability distribution and the spectral measure must also reproduce the expectation value of $\tilde{A}$, that is, $\braket{\tilde{A}}= \int_\RR a \braket{g_A(a),\rho}\dd a$, see Ref.~\cite{PlavalaKleinmann-oscillator}. It is sufficient to require that $\int_\RR g_A(a) \dd a = 1$ and $\int_\RR a g_A(a; \vec{q}, \vec{p}) \dd a = A(\vec{q}, \vec{p})$ while the positivity condition $\braket{g_A(I),\rho}\ge 0$ must hold for all states.

Comparing the definition of the spectral measure with the initial definition of the marginals of $\rho$, we identify, the spectral measure of the position observable $\vec{Q}(\vec{q},\vec{p})=\vec{q}$ to be
\begin{equation}
g_{\vec{Q}}(I;\vec{q},\vec{p})= \int_I \delta^{(3)}(\vec{q}_0-\vec{q}) \ddd{3} q_0
\end{equation}
and analogically for momentum. For other observables there is no unique way to define the phase-space spectral measure. Already the phase-space spectral measure for the energy of the harmonic oscillator differs between classical and quantum theory and other theories for the harmonic oscillator that are neither classical nor quantum have been proposed \cite{PlavalaKleinmann-oscillator}.

\section{Time-evolution}
We consider now the time evolution of a system in the sense that the random variable $\tilde{A}$ associated to an observable $A$ evolves over time. For the moment we assume that this time dependence can be solely attributed to a change of the state of the system,
\begin{equation}
\label{eq:time-Schrodinger}
\Pr[\,\tilde{A}(t)\in I\,] = \braket{g_A(I), \rho(t)}.
\end{equation}
The time evolution of the state is Markovian if, for any $t>0$ and any $t_0$, $\rho(t+t_0)$ is fully determined by $\rho(t_0)$ and $t$, that is,
\begin{equation}
\rho(t+t_0)= \mathcal{R}(t)\rho(t_0)
\end{equation}
for some linear map $\mathcal{R}(t)$ mapping states to states. Under mild assumptions such a Markovian time evolution yields the equation of motion $\dot{\rho}(t_0) = \mathcal{G}\rho(t_0)$, where $\mathcal{G}= \dot{\mathcal{R}}(0)$ is the generator of time shifts and $\mathcal R(t)= \e^{t\mathcal G}$. Here, we assume that $\mathcal{R}$ and $\mathcal{G}$ are linear maps in order to achieve that convex combinations are preserved by the time evolution, that is, $\rho(t+t_0) = p\rho_1(t+t_0)+(1-p)\rho_2(t+t_0)$ for $0\le p\le 1$.

The situation that we described so far corresponds to the Schrödinger picture in quantum theory, since only the state evolves in time. We switch now to the Heisenberg picture and assume that, equivalently, all of the time-dependence of $\tilde{A}(t)$ can be accounted for by the spectral measure while the state remains constant,
\begin{equation}
\label{eq:time-Heisenberg}
\Pr[\, \tilde{A}(t)\in I\,]= \braket{g_A(I;t), \rho}.
\end{equation}
Equating the time derivatives of the right hand sides of Eq.~\eqref{eq:time-Heisenberg} and Eq.~\eqref{eq:time-Schrodinger}, gives us $\braket{\dot{g}_A(I;t),\rho}= \braket{g_A(I),\dot{\rho}(t)}= \braket{g_A(I), \mathcal{G} \rho} = \braket{\mathcal{G}^\dag g_A(I),\rho}$, and hence the adjoint map $\mathcal{G}^\dag$ is the generator for the time shifts of observables.

In summary, we arrived at the dynamical equations
\begin{equation}
\label{eq:time-dynG}
\dot{\rho} = \mathcal{G} \rho
\quad \text{and} \quad
\dot{g}_A(I) = \mathcal{G}^\dagger g_A(I),
\end{equation}
where $\mathcal{G}$ and $\mathcal{G}^\dagger$ are linear maps with the latter being the adjoint map of the former with respect to $\braket{\cdot,\cdot}$.
It should be noted that in this formulation, $\rho(t)$ is required to be a 
state of the theory, for all times $t>0$, and similarly, $g_A(I;t)$ is a 
spectral measure of the theory for all $t>0$.

In quantum and classical theory, the generator $G$ is directly computed from 
the Hamiltonian $H(q,p)$ of the system. Since such a connection is crucial for 
a phase-space formulation we use a construction of $\mathcal G$ from $H$ that 
generalizes from both, the classical and quantum case.
We recall that the time evolution in classical mechanics is given by the 
Liouville equation $\dot{\rho} = \{H,\rho\}$. Here,
\begin{equation}
\{f,g\} = \sum_{i=1}^3 \left(\dfrac{\partial f}{\partial q_i} \dfrac{\partial g}{\partial p_i} - \dfrac{\partial f}{\partial p_i} \dfrac{\partial g}{\partial q_i}\right) = f D_\omega g
\end{equation}
denotes the Poisson bracket and we define the operator $D_\omega$ which acts both to the left and to the right. In contrast, using the Wigner function formalism for quantum theory \cite{Wigner-WignerFunctions,Case-wignerFunctions}, the time-evolution of a Wigner function is given as $\dot{\rho} = \mb{H,\rho}_{\QT}$, where
\begin{equation}
\mb{f, g}_{\QT} = \dfrac{2}{\hbar} f \sin\left( \frac{\hbar}{2} D_\omega
\right) g
\end{equation}
is the Moyal bracket \cite{Groenewold-QM,Moyal-WignerFunctions}. As a generalization of both, the Poisson bracket and the Moyal bracket, we define the generalized Moyal bracket as
\begin{equation}
\label{eq:time-generalizedMoyal}
\mb{f,g} = \{f,g\} + \sum_{n=1}^\infty a_n \hbar^{2n} f D_\omega^{2n+1} g,
\end{equation}
where $a_n$ are some dimensionless constants. We recover the Poisson bracket 
for $a_n = 0$ while $a_n = \frac{(-1)^n}{2^{2n} (2n+1)!}$ yields back the Moyal 
bracket. Moreover we will assume that, just like in quantum theory, $\hbar$ is 
a small constant and so the terms proportional to $\hbar^{2n}$ are microscopic 
corrections to the classical time evolution obtained in the limit $\hbar \to 
0$. The coefficients cannot be changed arbitrarily without changing other 
aspects of the theory. For example, if one applies arbitrary classical 
time-evolutions to Wigner functions from quantum theory, one readily obtains 
negative probabilities, rendering the theory inconsistent. Nevertheless, the 
coefficients $a_n$ are experimentally accessible in anharmonic systems. 
Consider the Hamiltonian $H(t)=\frac{p^2}{2m} + \frac{m\omega^2}2 q^2 
+\lambda(t)\frac{m^2\omega^3}{2\hbar} q^4$ where $\lambda(t)$ is a function of 
time. Then $\lambda(t)\frac{m^2\omega^3}{2\hbar} q^4$ will contribute terms 
proportional to $a_1$ to the time-evolution. Thus one can determine $a_1$ by 
introducing the anharmonic term and measuring $p^2$ at a later time, see 
Appendix~\ref{appendix:timedep} for the explicit calculation.

We summarize key properties of the generalized Moyal bracket.
\begin{enumerate}[(i)]
\item $\mb{\cdot,\cdot}$ is linear in both arguments.
\item $\mb{\cdot,\cdot}$ is antisymmetric, $\mb{f,g}=-\mb{g,f}$.
\item\label{item:time-GMB-polynomials} For $g$ a polynomial on phase space of at most second order and $f$ an arbitrary phase-space function, $\mb{f,g}=\{f,g\}$.
\item $\Lie_f: g\mapsto \mb{f,g}$ is skew-adjoint, $\Lie_f^\dagger = -\Lie_f$.
\item $\mb{\cdot,\cdot}$ satisfies the Jacobi identity if and only if it coincides with the Moyal bracket (up to the value of $\hbar$).
\end{enumerate}
The proof of Property~(i)--(iv) can be found in Appendix~\ref{appendix:propertiesGMB}; Property~(v) was proved in Ref.~\cite{Fletcher-moyalBracket}. We mention that antisymmetry and skew-adjointness follow from the fact that only odd powers of $D_\omega$ occur in the generalized bracket; including even powers would either contradict Markovianity of the time-evolution, or the generator of time-translations would be different from $\Lie_H$, see Appendix~\ref{appendix:propertiesGMB}. Other properties that are familiar from the Poisson bracket are not fulfilled by the Moyal bracket and the generalized bracket. In particular, $\Lie_f$ does neither obey Leibniz's rule $\Lie_f(gh)\neq h\Lie_fg + g\Lie_f h$ nor the chain rule $\Lie_f g(h) \neq g'(h) \Lie_f h$. Also there exist functions $f,g$ such that $\{f,g\} = 0$ but $\mb{f,g} \neq 0$, or such that $\mb{f,g} = 0$ but $\{f,g\} \neq 0$, see Appendix~\ref{appendix:propertiesGMB}.

We use the generator $\Lie_H: f \mapsto \mb{H,f}$ in order to refine the dynamical equations in Eq.~\eqref{eq:time-dynG}, that is,
\begin{equation} \label{eq:time-timeEvolution}
\dot{\rho} = \Lie_H \rho
\quad \text{and} \quad
\dot{g}_A(I) = -\Lie_H g_A(I).
\end{equation}
For time-independent Hamiltonians, these equations have the solutions
\begin{equation}
\dot{\rho}(t) = \e^{t \Lie_H} \rho
\quad \text{and} \quad
\dot{g}_A(I;t) = e^{-t \Lie_H} g_A(I),
\end{equation}
respectively. Note that Property~\eqref{item:time-GMB-polynomials} may hold at time $t_0$ but not at later times. For example, in order to obtain the time-evolution of position $q_i$ in the Heisenberg picture we have to solve the equation $\dot{Q}_i(t) = - \Lie_H Q_i(t)$ with the initial condition $Q_i(0; \vec{q}, \vec{p}) = q_i$ and in general $\Lie_H Q_i(0) = \{H, Q_i(0)\}$ does not imply that the same holds at later times.

\section{Properties of the time-evolution}
For a Markovian time-evolution, the generator $\mathcal G$ must be constant in time. This is consistent with our dynamical equations, since the generalized bracket is antisymmetric and thus $\dot{H} = - \Lie_H H = 0$. In particular, the energy is conserved on average, that is, $\frac{d}{dt}\braket{\tilde H}= \braket{\dot H,\rho} = 0$. However, this does not yet imply that the probability distribution of $\tilde H$ is constant, since $\dot{H} = 0$ does not establish $\dot{g}_H(I) = 0$. This property is further distinct from the question whether an eigenstate, that is, a state with sharp energy distribution, is constant in time. We mention that in classical and quantum theory, the distribution of $\tilde H$ is constant in time. For classical theory this follows at once from $g_H(E)=\delta(E-H)$ and $\dot g_H(I)=-\{H,g_H(I)\}= 0$ and in quantum theory from $[\hat H,\Pi(I)]=0$ where $I\mapsto \Pi(I)$ is the projection-valued spectral measure of the Hamilton operator $\hat H$. Similarly, any eigenstate in quantum theory is stationary, while this is not generally the case in classical mechanics or operational theories.

While it is possible to construct a spectral measure such that $\Lie_H g_H(I) 
\neq 0$, one can prove that for discrete energy spectrum and for states such 
that $\Pr[\tilde{H} = E_n] \neq 0$ only for at most two $n$ we have 
$\frac{d}{dt} \Pr[\tilde{H} \in I] = 0$, see 
Appendix~\ref{appendix:propertiesGMB}, Proposition~\ref{prop:propertiesGMB-max2Energy}. This has two 
immediate consequences: First, the probability distribution of $\tilde{H}$ is 
constant in time if the system is in an eigenstate. Second, observing the 
effects of the time-evolution of $g_H(I)$ requires the existence of a state 
with contributions to more than two energy levels. Since these contributions 
must not be merely due to convex combinations, such states are in a generalized 
superposition state \cite{AubrunLamiPalazuelosPlavala-superposition}.

A hallmark of quantum theory is Ehrenfest's theorem, which shows that the mean values of position and momentum follow classical equations of motion. We can establish an analogous result in our general framework. For this we assume a Hamiltonian of the form
\begin{equation}
H(\vec{q}, \vec{p})=T(\vec{p}) + V(\vec{q}).
\end{equation}
Then the observables of position $\vec{Q}(\vec{q},\vec{p})= \vec{q}$, momentum $\vec{P}(\vec{q},\vec{p})=\vec{p}$, velocity $\vec v(\vec{q},\vec{p})= \vec{\nabla}_p T(\vec{p})$, and force $\vec F(\vec{q}, \vec{p})=- \vec{\nabla}_q V(\vec{q})$, where $\vec{\nabla}_q$ and $\vec{\nabla}_p$ are the gradients in position and momentum respectively, obey classical equations of motion on average,
\begin{equation} \label{eq:ehrenfest}
\dfrac{d}{dt}\braket {\tilde Q_i} = \braket {\tilde v_i}
\quad \text{and} \quad
\dfrac{d}{dt}\braket {\tilde P_j} = \braket{\tilde F_j}.
\end{equation}
In order to see this, we first observe that $-\Lie_H\vec{q}= \vec v$ and $-\Lie_H \vec{p}=\vec F$, due to property~\eqref{item:time-GMB-polynomials} of the generalized bracket. Hence,
\begin{equation}
\begin{split}
\dot{\vec{Q}}(t) &= -\Lie_H \e^{-t \Lie_H}\vec{q} = \e^{-t \Lie_H}(-\Lie_H \vec{q}) \\
&= \e^{-t \Lie_H}\vec v = \vec{v}(t),
\end{split}
\end{equation}
and using analogous steps, one finds $\dot{\vec{P}}(t)= \e^{-t \Lie_H}\vec{F} = \vec{F}(t)$. Taking the mean value of these relations yields Eq.~\eqref{eq:ehrenfest}.

Despite of this result, the evolving observables $\vec{Q}(t;\vec{q},\vec{p})$ and $\vec{P}(t;\vec{q},\vec{p})$ cannot be generally understood as phase-space trajectories \cite{SteuernagelKakofengitisRitter-wignerTrajectories,OlivaKakofengitisSteuernagel-wignerTrajectories}. In particular, Liouville's theorem $\frac{d}{dt}\rho(t;\vec{Q}(t;\vec{q},\vec{p}), \vec{P}(t;\vec{q},\vec{p}))= 0$, does not apply already for the Moyal bracket. Similarly, for an observable $A$, in general we have $A(t;\vec{q},\vec{p})\ne A(\vec{Q}(t;\vec{q},\vec{p});\vec{P}(t;\vec{q},\vec{p}))$.

An important practical tool in quantum theory is the interaction picture. This is applicable for Hamiltonians of the form $H(t)=H_0+H_1(t)$, where $H_0$ is the free Hamiltonian and $H_1(t)$ is an interaction term which may depend explicitly on time. In the interaction picture, observables evolve according to the free Hamiltonian $H_0$ while the time-evolution of the state also involves the interaction $H_1(t)$. Accordingly, the state and any observable $A$ in the interaction picture are defined as
\begin{equation}
\rho_\mathrm{int}(t)= \e^{-t \Lie_{H_0}}\rho(t)
\quad \text{and}\quad
A_\mathrm{int}(t)= \e^{-t \Lie_{H_0}}A,
\end{equation}
respectively, where $\rho(t)$ and $A$ are understood in the Schrödinger picture with respect to the full Hamiltonian $H(t)$, that is, $\dot{\rho}(t) = \mb{H(t),\rho(t)}$. Consequently, in the interaction picture we have $\braket{\tilde{A}(t)} = \braket{A_\mathrm{int}(t),\rho_\mathrm{int}(t)}$ and the dynamical equation
\begin{equation}
\label{eq:nonstationary-timeEvolution}
\begin{split}
\dot\rho_\mathrm{int}(t) &=
-\Lie_{H_0} e^{-t \Lie_{H_0}} \rho(t) + \e^{-t \Lie_{H_0}} \Lie_{H(t)} \rho(t) \\
&= \left( -\Lie_{H_0} + \e^{-t \Lie_{H_0}} \Lie_{H(t)} \e^{t \Lie_{H_0}} \right) \rho_\mathrm{int}(t) \\
&= \e^{-t \Lie_{H_0}} \Lie_{H_1(t)} \e^{t \Lie_{H_0}} \rho_\mathrm{int}(t).
\end{split}
\end{equation}
Interestingly, if the generalized bracket satisfies the Jacobi identity, then Eq.~\eqref{eq:nonstationary-timeEvolution} simplifies to $\dot{\rho}_\mathrm{int}(t) = \mb{H_{1,\mathrm{int}}(t),\rho_\mathrm{int}(t)}$ and thus the time-evolution of $\rho_\mathrm{int}$ is generated by $H_1(t)$ in the interaction picture, $H_{1,\mathrm{int}}(t) = \e^{-t\Lie_{H_0}}H_\mathrm{int}(t)$. Although this can be a significantly more convenient expression, we argue that it does not constitute a sufficient argument to conclude that the generalized bracket has to satisfy the Jacobi identity.

\section{The hydrogen atom}
We now turn to a concrete physical system, the nonrelativistic hydrogen atom with the Hamiltonian
\begin{equation}
H = \dfrac{\abs{\vec{p}}^2}{2 \mu} - \dfrac{\kappa}{\abs{\vec{q}}}.
\end{equation}
Here, $\mu$ is the reduced mass of the electron and $\kappa = \frac{e^2}{4\pi\epsilon_0}$ with $e$ the elementary charge. In this notation, the Bohr radius is given by $a_0 = \frac{\hbar^2}{\kappa\mu}$.

Let us first consider the conserved quantities. In classical and quantum theory, besides energy, also the angular momentum $\vec L=\vec q\times \vec p$ and the Runge--Lentz vector $\vec{A} = \vec{p} \times \vec{L} - \mu \kappa \vec{q}\,\abs{\vec{q}}^{-1}$ are constants in time. This turns out to be also the case for the generalized bracket, as can be verified in a straightforward calculation.

\begin{figure}
\includegraphics[width=\linewidth]{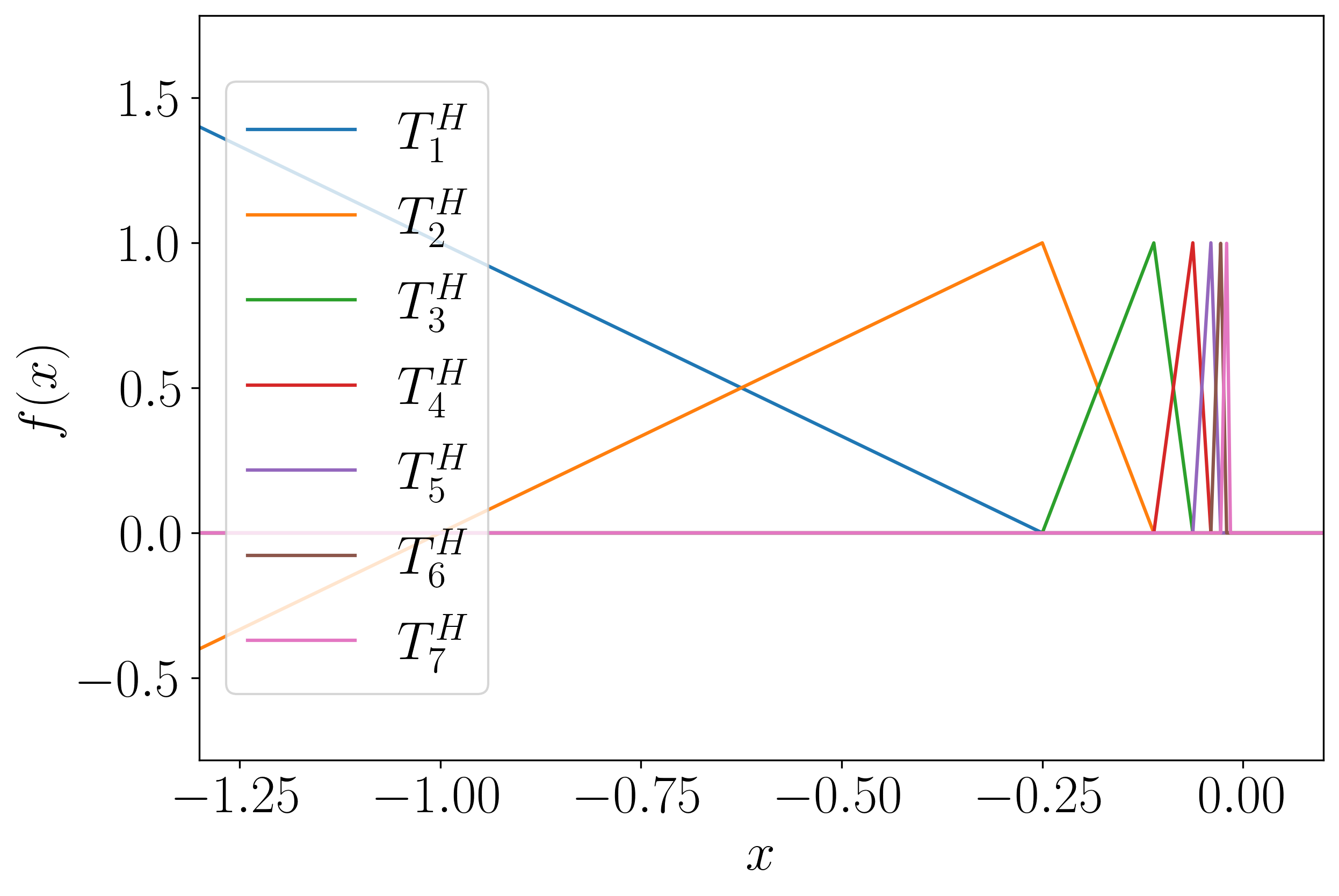}
\caption{The functions $T^H_n(x)$ used for the spectral measure of the energy of the hydrogen atom in Eq.~\eqref{eq:hydrogen-PSSM}. The functions are plotted for $n = 1,2, \ldots, 7$ and $x$ is in units of $\frac{\kappa}{2 a_0}$.}
\label{fig:hydrogen-TH}
\end{figure}

Quantum theory can explain the stability of the hydrogen atom as well as its 
spectrum which is in contrast to classical theory. But these two properties are 
by far not special for quantum theory per se. In the following we construct a 
phase space spectral measure $g_H(I)$ for the energy observable that is 
different from quantum theory, but has similar features. We first fix the 
energy spectrum to the one known from nonrelativistic quantum theory,
\begin{equation} \label{eq:hydrogen-En}
E_n = - \dfrac{\kappa}{2 a_0} \dfrac{1}{n^2}
\end{equation}
where $n = 1,2,\ldots$ is the principal quantum number. For any set of negative energies $I$ we then define the phase space spectral measure as
\begin{equation} \label{eq:hydrogen-PSSM}
g_H(I; \vec{q}, \vec{p}) = \sum_{n\colon E_n\in I} T^H_n(H(\vec{q},\vec{p}))
\end{equation}
where $T^H_n$ are the functions depicted in Figure~\ref{fig:hydrogen-TH}, see Appendix~\ref{appendix:hydrogen} for details. This measure is the uniquely specified by the following conditions:
\begin{enumerate}[(i)]
\item $g_H(I; \vec{q}, \vec{p})$ is piecewise linear function of $H(\vec{q}, \vec{p})$.
\item $g_H(E_n; \vec{q}, \vec{p}) \neq 0$ only if $H(\vec{q}, \vec{p}) \in [E_{n-1}, E_{n+1}]$.
\end{enumerate}
The second condition means that only the region of phase-space ``close'' to $H(\vec{q}, \vec{p}) = E_n$ contributes to the probability of observing the particle in the $n$\textsuperscript{th} energy level. This is a straightforward natural choice.

\begin{figure}
\begin{subfigure}{\linewidth}
\includegraphics[width=\linewidth]{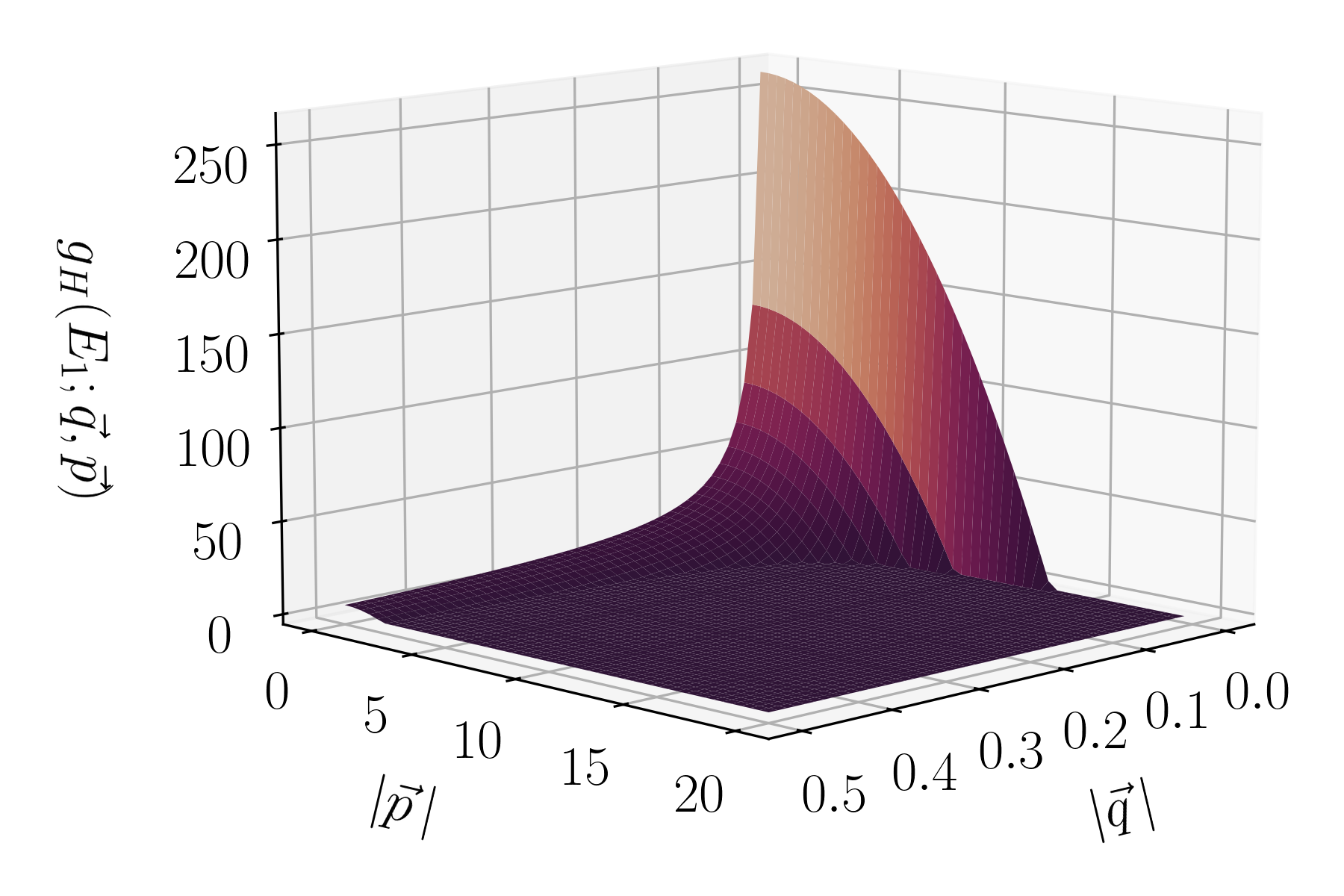}
\caption{Spectral function $g_H(E_1,\vec{q}, \vec{p})$.}
\end{subfigure}
\begin{subfigure}{\linewidth}
\includegraphics[width=\linewidth]{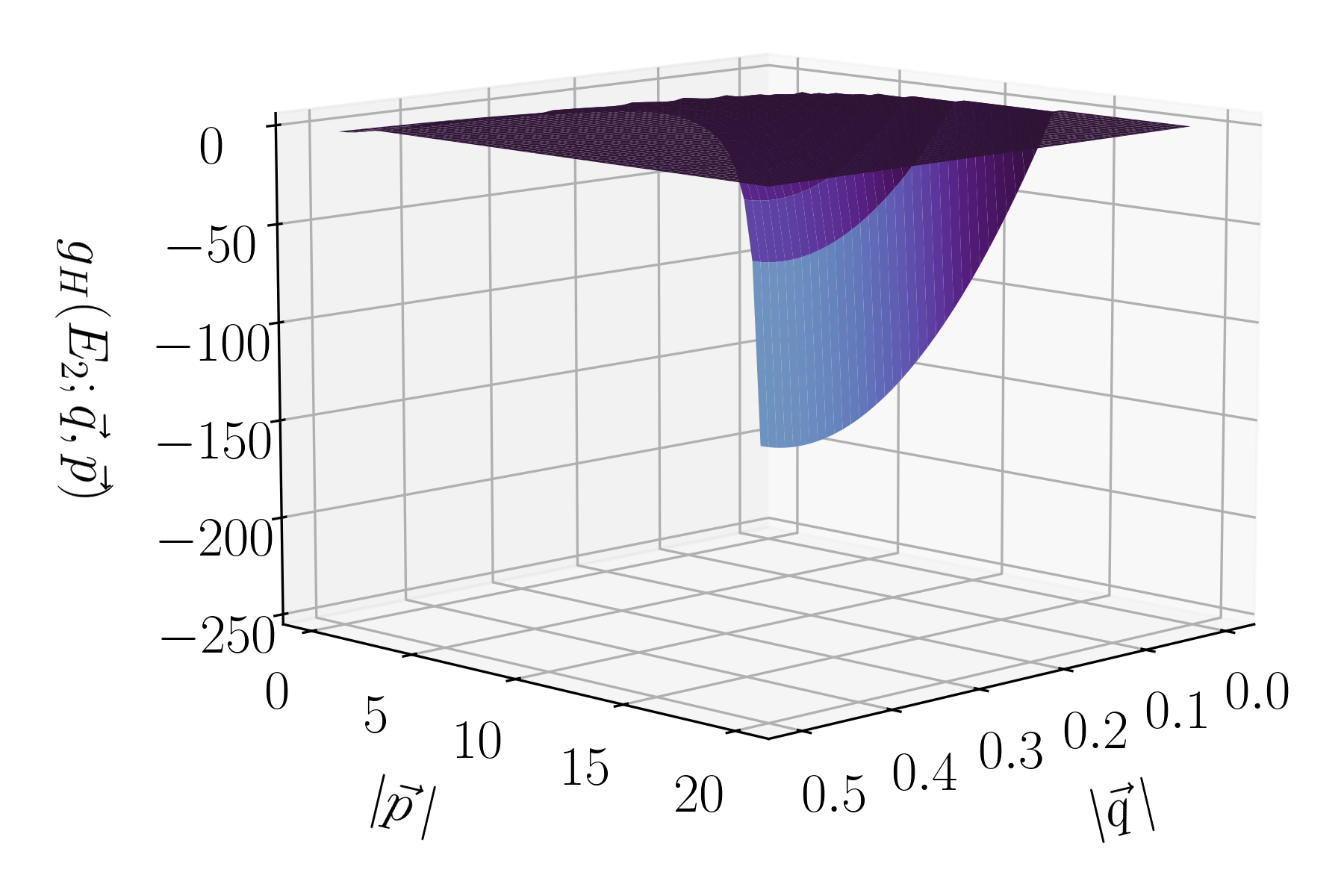}
\caption{Spectral function $g_H(E_2,\vec{q}, \vec{p})$.}
\end{subfigure}
\caption{Spectral functions for the first and second energy level. The Hamilton operator $H(\vec q,\vec p)$ is decomposed into spectral functions $g_H(E_k;\vec q,\vec p)$, such that $\braket{g_H(E_k)}$ is the probability to observe the energy $E_k$. The upper (lower) panel shows the spectral function for $E_1$ ($E_2$) as a function of $\abs{\vec{q}}$ and $\abs{\vec{p}}$ in units of the Bohr radius $a_0$ and $\hbar / a_0$, respectively. The normalization enforces that for small $\abs{\vec q}$ and $\abs{\vec p}$ both functions have to sum to unity. Here the spectral function for $E_2$ is strongly negative, eventually preventing the collapse of the atom.}
\label{fig:hydrogen-gHE}
\end{figure}

It follows from the definition of $T^H_2$ that $g_H(E_2; \vec{q}, \vec{p}) < 0$ everywhere in phase-space where $H(\vec{q}, \vec{p}) < E_1$. From this we can conclude that the hydrogen atom must be stable in the sense that no state can be supported solely where $H(\vec{q}, \vec{p}) \le E_1-\epsilon$ with $\epsilon>0$, see also Figure~\ref{fig:hydrogen-gHE}. Indeed, if $\supp(\rho)$ would be the support of such a state, then the continuity of $g_H(E_2;\vec q,\vec p)$ implies that $M = \sup\set{H(\vec q,\vec p) : (\vec q,\vec p) \in \supp(\rho) } < 0$ and hence
\begin{equation}
\begin{split}
\braket{g_H(E_2),\rho} &= \int_{\supp(\rho)} g_H(E_2;\vec q,\vec p)\rho(\vec q,\vec p) \ddd3q \ddd3p \\
&\le M \int_{\supp(\rho)} \rho(\vec q,\vec p) \ddd{3}q \ddd{3}p = M < 0,
\end{split}
\end{equation}
which is a contradiction with our requirement that all probabilities must be positive.

The measure $g_H(I; \vec{q}, \vec{p})$ is not stationary, $\Lie_H g_H(I) \neq 0$, because of discontinuities in the derivatives of $g_H(I)$. But we have $\{H, g_H(I)\} = 0$, so the measure is stationary if we approximate the time-evolution by the Poisson bracket, up to order $\hbar^2$. In a similar way one can construct the observable of angular momentum. For the $i$th component of $\vec{L}$ we let
\begin{equation}
g_{L_i}(I; \vec{q}, \vec{p}) = \sum_{m \colon m\hbar \in I} T^{L_i}_m(L_i(\vec{q}, \vec{p}))
\end{equation}
where $T^{L_i}_m$ are the functions depicted in Figure~\ref{fig:hydrogen-TLz} and defined in Appendix~\ref{appendix:hydrogen}.

\begin{figure}
\includegraphics[width=\linewidth]{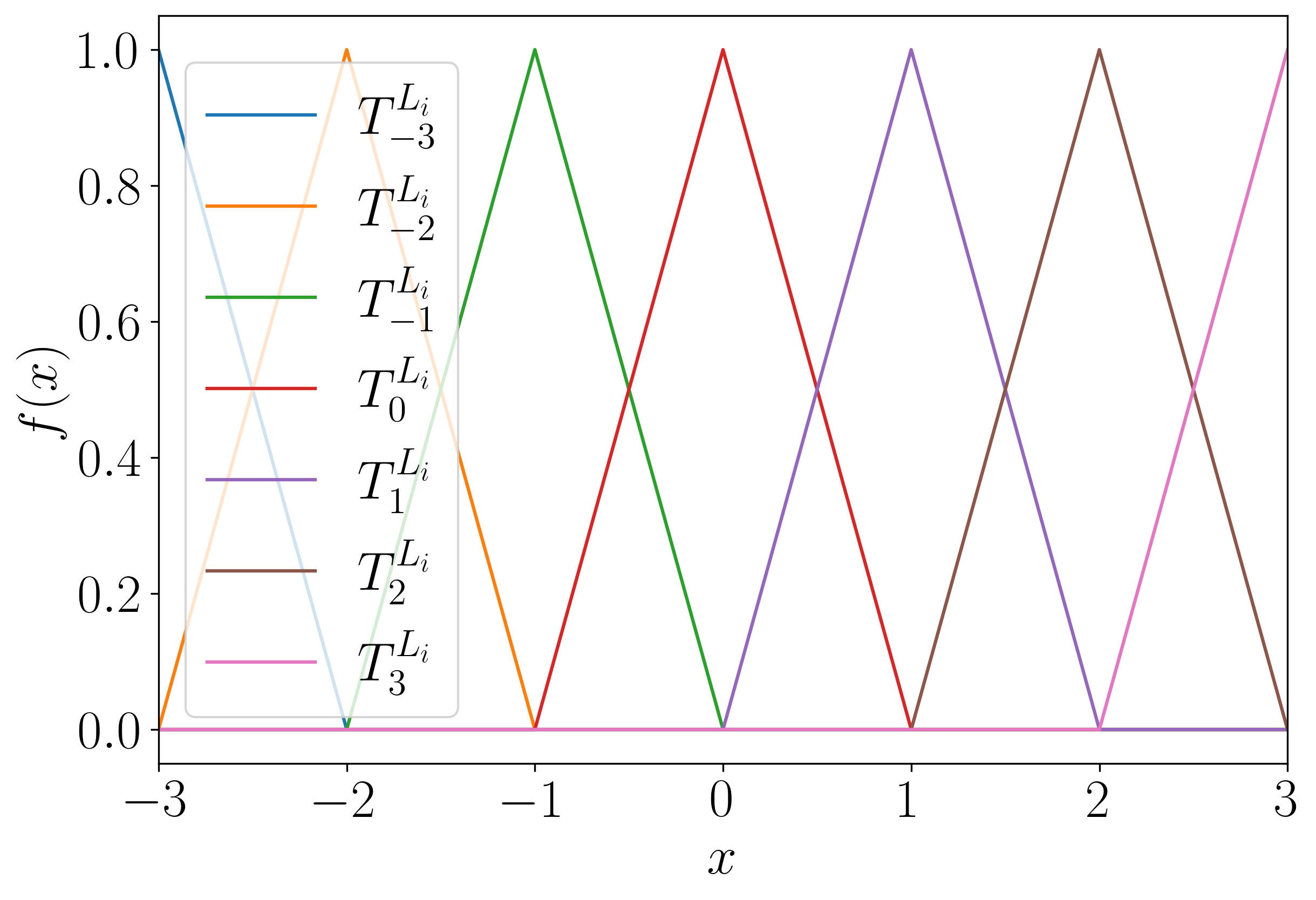}
\caption{The functions $T_m^{L_i}$ used in the definition of the phase space spectral measure of the angular momentum. The functions are plotted for $m = -3, \ldots, 3$ and $x$ is in units of $\hbar$.}
\label{fig:hydrogen-TLz}
\end{figure}

For possible states of the hydrogen atom we consider the Gaussian distribution
\begin{equation} \label{eq:hydrogen-state-Gaussian}
\rho_G(\vec{q},\vec{p})= \dfrac{1}{(2\pi)^3 \sigma_q^3 \sigma_p^3} \e^{-\frac{\abs{\vec{q}}^2}{2\sigma_q^2}-\frac{\abs{\vec{p}}^2}{2\sigma_p^2}}.
\end{equation}
Not all choices of $\sigma_p$ and $\sigma_q$ give a valid state, since $\braket{g_H(E_2),\rho_G}$ can be negative, in particular if $\sigma_q$ is too small. One can compute the value of $\braket{g_H(E_2),\rho_G}$ numerically for different pairs of $\sigma_q$ and $\sigma_p$ to find a region where $\rho_G$ gives positive probabilities, see Figure~\ref{fig:hydrogen-state-UR}. Note that we verify the positivity condition only at time $t = 0$. In principle we should verify whether $\braket{g_H(E_n),\rho(t)}$ is nonnegative for all later times. For the Poisson bracket, this is not the case because Leibniz rule renders the spectral measure to be constant in time by virtue of the Leibniz rule. It follows that $\braket{g_H(E_n),\rho}<0$ at $t>0$ must be at least of order $\hbar^2$. The same consideration also holds for the positivity of the marginals $\rho_q$ and $\rho_p$.

For $\sigma_p \to 0$ and $\sigma_q=\sigma_{\gnd} \approx 1.59577048804 \, a_0$ the state $\rho_G$ becomes
\begin{equation} \label{eq:hydrogen-state-deltaP}
\rho_{\gnd} (\vec{q},\vec{p})= \dfrac{1}{(2\pi)^{\frac{3}{2}} \sigma_{\gnd}^3} \e^{-\frac{\abs{\vec{q}}^2}{2\sigma_{\gnd}^2}} \delta^{(3)}(\vec{p}).
\end{equation}
This state is is a good approximation of a ground state, having $\braket{g_H(E_2),\rho_{\gnd}}=0$ and $\braket{\tilde H}\approx (1-10^{-6})E_1$ and subsequently we use it as if it was a proper ground state. We mention that $\rho_\gnd$ is clearly incompatible with quantum theory because the preparation uncertainty relation of position and momentum is violated.

For the ground state one obtains the most probable distance from the center at 
$t=0$ as $\sqrt{2} \sigma_Z \approx 2.26a_0$. But the state is not stationary 
since already $\{H,\rho_{\gnd}\}\ne 0$ and, for example, the most probable 
distance changes with time. However, the energy distribution $\Pr[\tilde{H} \in 
I]$ is constant in time since the state is an eigenstate, see 
Appendix~\ref{appendix:propertiesGMB}, 
Proposition~\ref{prop:propertiesGMB-max2Energy}. The consistency of the model 
can be extended beyond the ground state, by considering any nonnegative state 
$\rho(\vec q,\vec p)\ge 0$ with $\braket{g_H(E_2),\rho}\ge 0$ and for the 
time-evolution given by the Poisson bracket, that is, all $a_i=0$. Since 
$\{H,g_H(E_n)\}=0$, for any such state $\Pr[\tilde{H} \in I]$ is constant in time and remains  
positive for all other nonnegative spectral measures, since the classical 
time-evolution preserves positivity of phase-space functions. This is the case 
for rotations $\exp(\mathcal L_{\vec L\cdot \vec v})$, due to $\{\vec L, 
g_H(E_n)\}=\vec 0$, but is generally not true for other transformations, such as for time evolution given by harmonic potential.

\begin{figure}
\includegraphics[width=\linewidth]{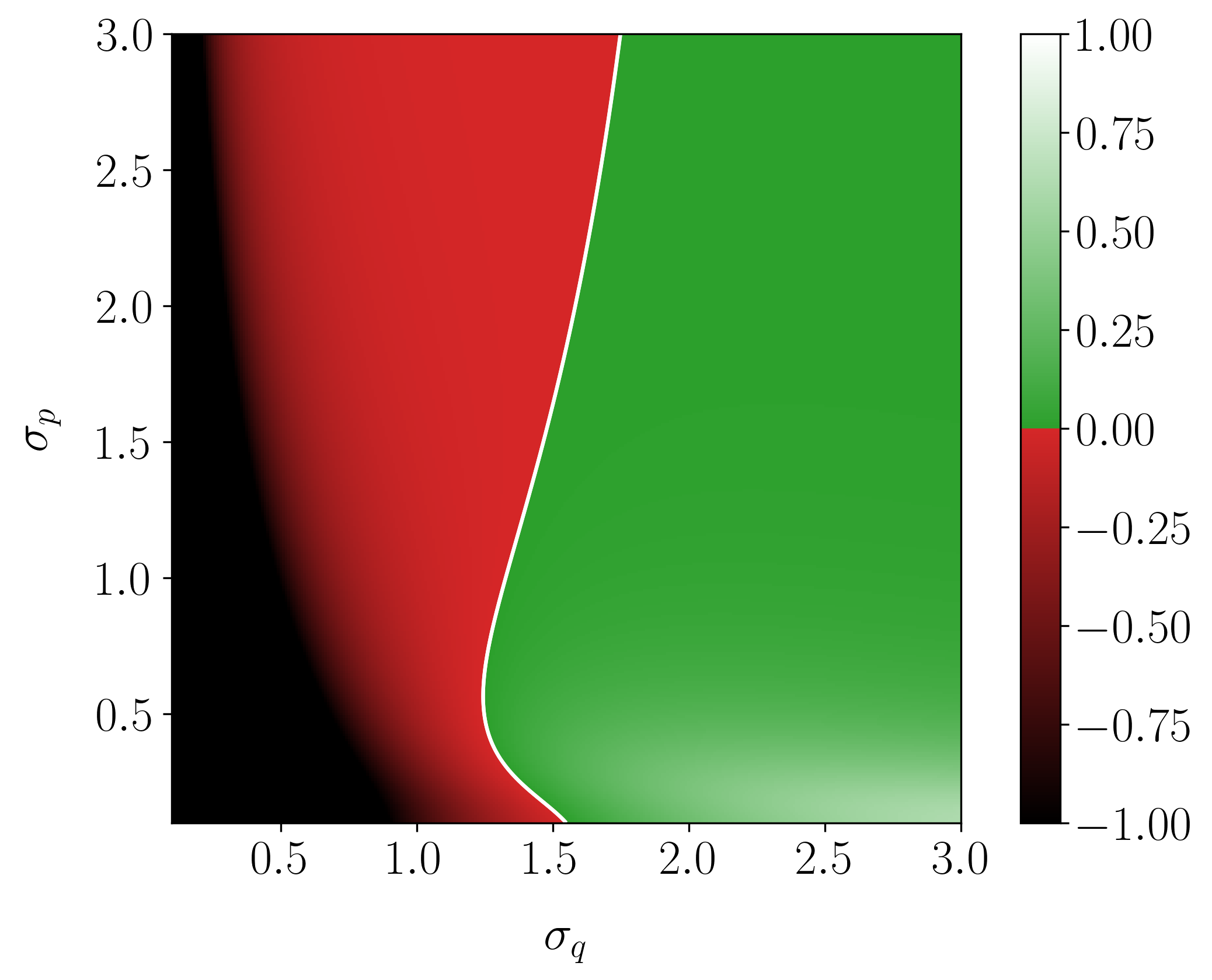}
\caption{Value of $\braket{g_H(E_2),\rho_G}$ for the family of states $\rho_G$ given by Eq.~\eqref{eq:hydrogen-state-Gaussian}. For a valid state this value has to be nonnegative. The positivity of $\braket{g_H(E_2),\rho_G}$ is significantly influenced by $\sigma_q$, therefore sufficient position uncertainty is necessary. $\sigma_q$ is in units of $a_0$ and $\sigma_p$ is in units of $\hbar/a_0$. The numerical uncertainty of the calculation is equal to the width of the white line separating the positive and negative region.}
\label{fig:hydrogen-state-UR}
\end{figure}

\section{External magnetic field}
We briefly treat the influence of an external magnetic field to our model of a hydrogen atom in our model of a hydrogen atom. For simplicity we assume that the electron is spinless and so the external magnetic field only interacts with the angular momentum. Hence, for a homogeneous magnetic field in $z$-direction, the new Hamiltonian is
\begin{equation}
H_B(\vec{q}, \vec{p}) = H(\vec{q}, \vec{p}) + \frac{\mu_B}{\hbar} B L_3(\vec{q}, \vec{p}),
\end{equation}
where $\mu_B$ denotes the Bohr magneton. We construct the phase space spectral measure for $H_B(\vec{q}, \vec{p})$ using a product of functions. In quantum theory we would get similar result using the Moyal product \cite{Moyal-WignerFunctions,Fairlie-starProduct,Zachos-starProduct}, but for our purposes, the ordinary point-wise product is sufficient. We get
\begin{equation}\label{eq:magnetic-pssmhb}
g_{H_B} (I) = \sum_{n,m\colon E_{n,m}\in I} T_n^H(H) T_m^{L_3}(L_3)
\end{equation}
with $E_{n,m}=E_n+\mu_B Bm$ and $m$ the magnetic quantum number. Thus we get a splitting of the energy levels due to the external magnetic field.

At first sight it may seem that there is no restriction of the magnetic number $m$, quite in contrast to quantum theory where $\abs{m}<n$ holds. However, certain combinations of quantum numbers cannot occur according to Eq.~\eqref{eq:magnetic-pssmhb}, simply because $T_n^H(H)$ and $T_m^{L_3}(L_3)$ have disjoint support in phase space. Using this we find that $\Pr[\, \tilde H_B=E_{n,m}\,] > 0$ only if $\abs{m} \leq 2(n+1)$, see Appendix~\ref{appendix:hydrogen} for a proof.

\section{Non-stationary external electric field}
We now investigate the case when the atom is perturbed by a non-stationary external electric field, in particular, by an external electromagnetic wave. We take into account only the electric field corresponding to the electromagnetic wave and we assume that the wavelength of the electromagnetic wave is significantly larger than the size of the atom. Thus the new Hamiltonian is
\begin{equation}
H_E(\vec{q}, \vec{p}, t) = H(\vec{q}, \vec{p}) + H_e(\vec{q}, \vec{p}, t)
\end{equation}
where
\begin{equation}
H_e(\vec{q}, \vec{p}, t) = - 2 e E \sin (\omega t) q_3,
\end{equation}
assuming that the electric field is oriented in the $z$-direction.

We proceed in the interaction picture as discussed above. Expanding the exponentials in Eq.~\eqref{eq:nonstationary-timeEvolution} we get
\begin{equation}
\dfrac{d}{dt} \rho_\mathrm{int} = \Lie_{H_e} \rho_\mathrm{int} + t ( \Lie_{H_e} \Lie_{H} - \Lie_{H} \Lie_{H_e}) \rho_\mathrm{int} + \dotsm,
\end{equation}
where we omitted terms which are second order in $t \Lie_H$ and higher. One finds $\Lie_{H_e} \Lie_{H} - \Lie_{H} \Lie_{H_e} = \Lie_{G}$ with
\begin{equation}
G = - \dfrac{2eE}{\mu} \sin(\omega t) p_3,
\end{equation}
see Appendix~\ref{appendix:electric}. Thus we get the approximation
\begin{equation}
\dfrac{d}{dt} \rho_\mathrm{int} \approx \Lie_{(H_e + t G)} \rho_\mathrm{int}.
\label{eq:nonstationary-finalTimeEvolution}
\end{equation}
Since the effective Hamiltonian $H_{\eff} = H_e + t G$ is only a linear function of position and momentum, the generalized bracket reduces to the Poisson bracket. Then $\rho(t;\vec q,\vec p)=\rho(0;\vec q+\vec s(t);\vec p+\vec u(t))$ is a solution of Eq.~\eqref{eq:nonstationary-finalTimeEvolution} for
\begin{equation}
\begin{split}
\vec s(t)&=\int_0^t \{ H_{\eff}(\tau),\vec q\}\dd\tau\\&
=\dfrac{2 e E}{\mu \omega^2} ( \sin(\omega t) - \omega t \cos(\omega t) )\vec e_3,\\
\vec u(t)&=\int_0^t \{H_{\eff}(\tau),\vec p\}\dd\tau=\dfrac{2eE}{\omega}(\cos(\omega t)- 1)\vec e_3.
\end{split}
\end{equation}
It is now straightforward to compute $\Pr[\tilde{H}(t) = E_n]$ numerically. Hereby the phase space spectral measure in the interaction picture depends on time, however this introduces only corrections of the order $\hbar^2$, which we neglect. The results of our computations are plotted in Figure~\ref{fig:nonstatioanry-compareQT} for $\omega = (E_2 - E_1)/\hbar$ and $\rho_{\gnd}$ as initial state. Importantly, we see that there is a non-zero probability of exciting the atom to the second energy level.

\begin{figure}
\includegraphics[width=\linewidth]{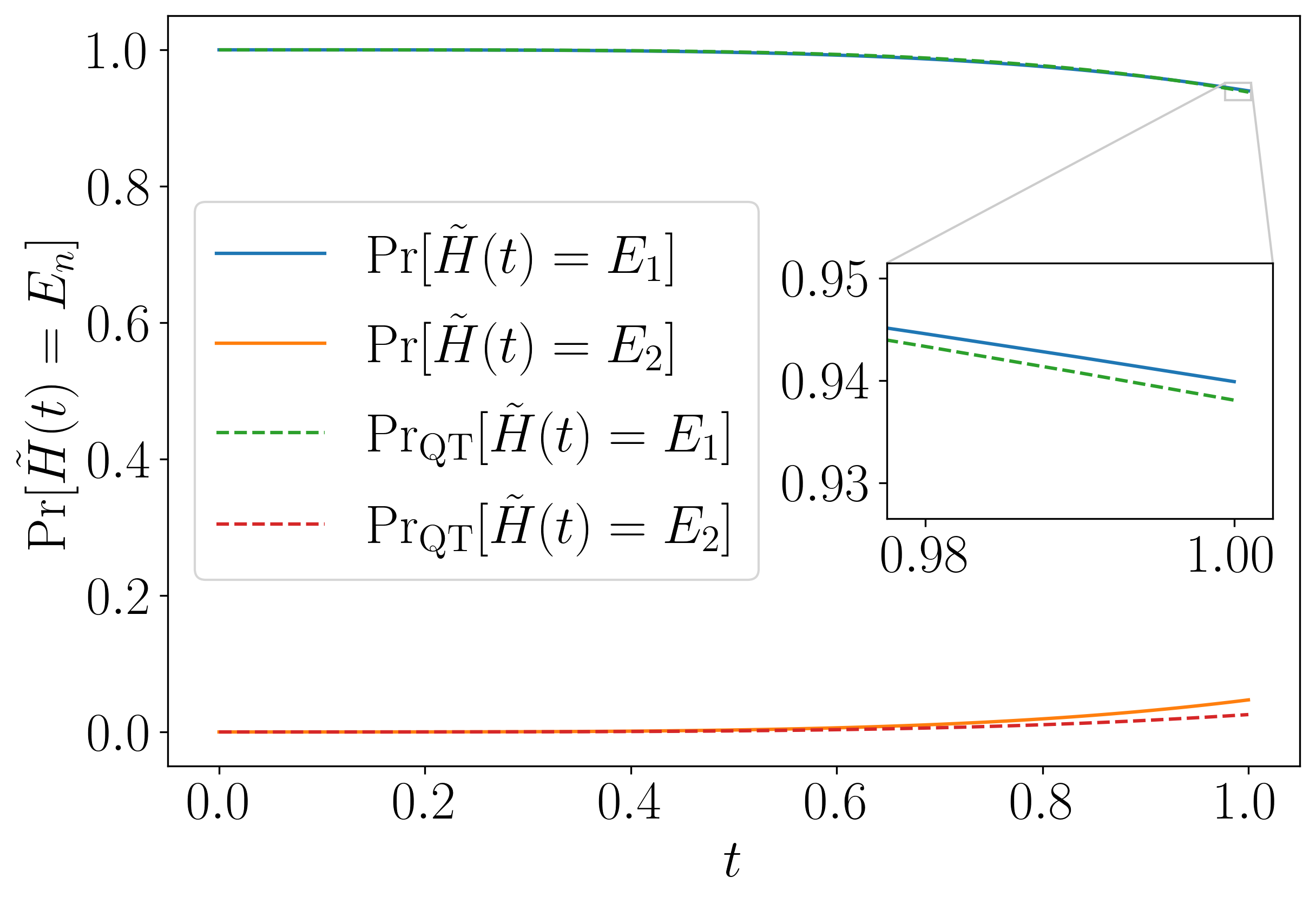}
\caption{Time evolution of the probabilities $\Pr[\tilde{H}(t) = E_1]$ and $\Pr[\tilde{H}(t) = E_2]$ for the initial state $\rho_{\gnd}$ of our model of the hydrogen atom and of the probabilities $\Pr_\QT[\tilde{H}(t) = E_1]$ and $\Pr_\QT[\tilde{H}(t) = E_2]$ for initial state $\psi_{100}$ of the quantum hydrogen atom, both in the same approximation. Notice that $\Pr[\tilde{H}(t) = E_1] \approx \Pr_\QT[\tilde{H}(t) = E_1]$ up to the numerical precision of our calculations. Here $\omega = (E_2 - E_1)/\hbar$ and $eE = \mu a_0 \omega^2$ and $t$ is in units of $\frac{\hbar}{\abs{E_1}}$.}
\label{fig:nonstatioanry-compareQT}
\end{figure}

For comparison, we also compute the transition probability $\Pr_\QT[\tilde{H}(t) = E_n]$ in quantum theory. We use the same approximation and assume as initial state $\rho$ the ground state of the hydrogen atom with the wave function $\psi_{100}$. The Schrödinger equation with Hamiltonian $\hat H_{\eff}$ yields
\begin{equation}
\psi_{100}(t;\vec x)= \e^{i\phi(t)}\e^{-\frac i\hbar \vec x\cdot \vec u(t)} \psi_{100}(0;\vec x+\vec s(t)),
\end{equation}
where $\phi(t)$ is an irrelevant phase. The results are plotted in Figure~\ref{fig:nonstatioanry-compareQT}. We see that also in quantum theory there is a nonzero probability of exciting the atom in our approximation and $\Pr[\tilde{H}(t) = E_1] \approx \Pr_\QT[\tilde{H}(t) = E_1]$ up to the numerical precision of our calculations. We finally mention that the excitation of the first energy level is of different origin than Rabi oscillations in quantum optics. This is because we consider here the limit of short times while in quantum optics one usually applies the rotating-wave approximation which disregards fast oscillations.

\section{Scattering theory}
In both classical and quantum theory the scattering of charged particles on a Coulomb potential leads to the same differential cross section given by the Rutherford formula \cite{Rutherford-scattering}
\begin{equation}
\dfrac{d \sigma}{d \Omega} = \dfrac{\kappa^2 \mu^2}{4 p_0^4} \dfrac{1}{\sin^4(\vartheta / 2)}
\end{equation}
where $p_0$ is the momentum of the incoming particles and $\vartheta$ is the 
scattering angle. In this section we show that this result holds in all 
operational theories of the hydrogen atom where the time-evolution is described 
by the generalized Moyal bracket \eqref{eq:time-generalizedMoyal}. Moreover 
this result is independent of the phase space spectral measure of the energy 
observable, since we only use the Hamiltonian as the generator of 
time-evolution.

The Wigner function representation was used before to investigate the scattering \cite{Remler-WignerFunctionsScattering,CarruthersZachariasen-WignerFunctionsScattering,KarlovetsSerbo-WignerFunctionsScattering} and our approach is in particular based on Ref.~\cite{CarruthersZachariasen-WignerFunctionsScattering}. We assume that in the asymptotic past, $t\to -\infty$, the scattering particles have a uniform spacial density $\nu$ and a fixed momentum $\vec p_0$. Dropping the normalization condition of the state, we write
\begin{equation}
\rho_{\ini}(\vec{q}, \vec{p})= \lim_{t\to-\infty}\rho(t;\vec q,\vec p)= \nu\delta^{(3)} (\vec{p} - \vec{p}_0).
\end{equation}
Note that this is the same initial condition as one uses in quantum scattering theory. The particle density at later times is given by
\begin{equation}
\label{eq:scattering-density}
D(t;\vec{q}) = \int_{\RR^3} \rho(t; \vec{q}, \vec{p}) \ddd{3} p,
\end{equation}
and for computing the cross-section we aim to obtain this density in the far field for the asymptotic future, that is, for $t\to +\infty$ and $\abs{\vec q}\to \infty$.

Using a Green's functions approach \cite{CarruthersZachariasen-WignerFunctionsScattering}, the formal solution of the dynamical equations in Eq.~\eqref{eq:time-timeEvolution} is given by
\begin{widetext}
\begin{equation}
\label{eq:scattering-solution}
\rho(t; \vec{q}, \vec{p}) = \rho_{\ini}(t; \vec{q}, \vec{p}) + \int\limits_{\RR^3} \int\limits_{-\infty}^t K(\vec{q} - \dfrac{\vec{p}}{\mu}(t - \tau), \vec{p}, \vec{p'}) \rho(\tau; \vec{q} - \dfrac{\vec{p}}{\mu}(t - \tau), \vec{p'}) \dd \tau \ddd{3} p'
\end{equation}
\end{widetext}
where
\begin{equation}
\label{eq:scattering-K}
K(\vec{q}, \vec{p}, \vec{p'}) = \mb{ V(\vec{q}), \delta^{(3)}_p(\vec{p} - \vec{p'})},
\end{equation}
see Appendix~\ref{appendix:scattering} for the full derivation. The solution can be found in a perturbative manner via $V(\vec{q}) \mapsto \lambda V(\vec{q})$ and the expansion
\begin{equation} \label{eq:scattering-rhoSeries}
\rho(t; \vec{q}, \vec{p}) = \sum_{n=0}^\infty\sum_{k=0}^\infty \hbar^{2n} \lambda^k \rho_{n,k}(t; \vec{q}, \vec{p}).
\end{equation}
by comparing coefficients in $\hbar$ and $\lambda$ in Eq.~\eqref{eq:scattering-K}.

Using these techniques, we show in Appendix~\ref{appendix:scattering} that
\begin{equation}
D_{n,k}(t;\vec{q}) = \int_{\RR^3} \rho_{n,k}(t;\vec q,\vec p)\ddd{3}p = \dfrac{f_{n,k}(t;\vartheta)}{\abs{\vec{q}}^{2n+k}},
\end{equation}
for some functions $f_{k,n}(t;\vartheta)$. The only terms that can contribute to the far field differential cross section are now for $n = 0$ and $k = 0, 1, 2$ and $n = 1$ and $k = 0$. The terms with $n = 0$ are the classical terms which are obtained by replacing the generalized bracket in Eq.~\eqref{eq:scattering-K} by the Poisson bracket and hence they give us the same prediction for the differential cross section as classical theory. We never get terms of the order $\hbar^{2} \lambda^0$ on the right-hand-side of Eq.~\eqref{eq:scattering-solution} because the terms of order $\lambda^0$ are the terms that do not include the potential but only the initial state $\rho_{\ini}$. Thus the differential cross section must be given by Rutherford's formula in all operational theories where the time-evolution is given by the generalized Moyal bracket.

\section{Conclusions}
We constructed a toy model for the hydrogen, which does not fall into the 
formalism of quantum theory or classical theory and as such clearly does not 
satisfy the constraints from modern experimental data. But, as we showed here, 
this model is internally consistent and conceptually it is not obvious why the 
toy model is incorrect, other than that it is not a quantum model. Moreover the 
toy model is in accordance with experimental predictions and theoretical 
paradigms of early quantum theory: The collapse of the atom is prevented due to 
the uncertainty principle, the model features a discrete energy spectrum in 
accordance with experimental observations, the angular momentum is quantized as 
predicted by Bohr, perturbations by nonstationary electric field lead to 
excitations, and Rutherford's formula for scattering cross-section holds. 
Moreover there is a meaningful classical limit: $\hbar \to 0$ recovers the 
classical dynamics given by the Poisson bracket and localized particles distant 
enough from the center of the potential are possible within our toy model.

After having established a formal background for dynamical theories on phase space, a conceptually rather straightforward construction gives already our toy model. The model does not even remotely resemble quantum theory, but still makes precise, measurable predictions. Besides that this shows how little is known about the space of theories in which quantum theory resides as a special case, we also found evidence that quantum theory is a strikingly simple theory with curious mathematical coincidences. For example, since the Moyal bracket satisfies the Jacobi identity, the interaction picture is especially simple to handle. We found that although conservation of energy always holds on average, the distribution of the energy can change over time. In quantum theory and classical theory the energy distribution is constant roughly because in both theories $H$ and $H^2$ ``commute.''

Generally, while it clearly is possible to identify experiments that invalidate our toy model in favor of quantum theory, it is an upcoming theoretical challenge to find operational postulates that are obeyed by quantum theory but violated in the toy model. While such postulates are known for finite-dimensional systems, it is not straightforward to generalize them to continuous variable systems.

\begin{acknowledgments}
We acknowledge support from the Deutsche Forschungsgemeinschaft (DFG, German Research Foundation, project numbers 447948357 and 440958198), the Sino-German Center for Research Promotion (Project M-0294), the ERC (Consolidator Grant 683107/TempoQ), and the German Ministry of Education and Research (Project QuKuK, BMBF Grant No. 16KIS1618K).
M.P. is thankful for the financial support from the Alexander von Humboldt Foundation.
The OMNI cluster of the University of Siegen was used for the computations.
\end{acknowledgments}

\bibliography{citation}

\onecolumngrid
\appendix

\section{Experimental accessibility of the coefficients of the generalized bracket} \label{appendix:timedep}
The generalized bracket $\mb{\cdot,\cdot}$ is defined in term of dimensionless coefficients $a_n$, see Eq.~\eqref{eq:time-generalizedMoyal}. In quantum theory these coefficients take specific values, for example, $a_1=-\frac1{24}$. One may take the perspective that each of these coefficients is a constant of nature and therefore must be verified in an experiment. However, the coefficients are not accessible by measuring the spectrum of an observable, because this spectrum can be chosen freely, as we have seen for the hydrogen atom. But the coefficients do become experimentally accessible in a time-dependent potential. In order to obtain $a_1$, a possible way is to determine the variance cf the momentum for the time-dependent Hamiltonian
\begin{equation} \label{eq:timedep-anharmonic}
H(t)=\dfrac{p^2}{2m} + \dfrac{m\omega^2}2 q^2 +\lambda(t)\dfrac{m^2\omega^3}{2\hbar} q^4.
\end{equation}
Importantly, this Hamiltonian does not ``commute'' with itself at different times, that is, $\mb{H(t),H(t')}\ne 0$. Hence the formal solution of the dynamical equation for an observable $A$ is not simply given by $A(t+t_0)=\exp(-\Lie_Ht)A(t_0)$ but rather by the time-ordered exponential, which can be approximated as
\begin{equation}
A(t+t_0) \approx
(\mathrm{id}-\tfrac tn\Lie_{H(t_{n-1})})
(\mathrm{id}-\tfrac tn\Lie_{H(t_{n-2})}) \dotsm
(\mathrm{id}-\tfrac tn\Lie_{H(t_{0})})A(t_0),
\end{equation}
with $t_k=\frac knt+t_0$ and $n \in \NN$. For $A(t_0)=p^2$, the first dependence on $a_1$ occurs at the fifth step, with
\begin{equation}
(\mathrm{id}-\tfrac tn\Lie_{H(t_{4})}) \dotsm (\mathrm{id}-\tfrac tn\Lie_{H(t_{0})})p^2 = K +\frac{1728}{n^4} m^2\omega^6 \, a_1 \lambda(t_1)\lambda(t_4) q^2 t^4,
\end{equation}
assuming that $\lambda(t_0)=0$ and where $K$ is a term independent of $a_1$, that is, corresponding to the case $a_1 = 0$. Hence it is possible to determine $a_1$ by measuring $p^2$.

\begin{figure}
\includegraphics[width=.4\linewidth]{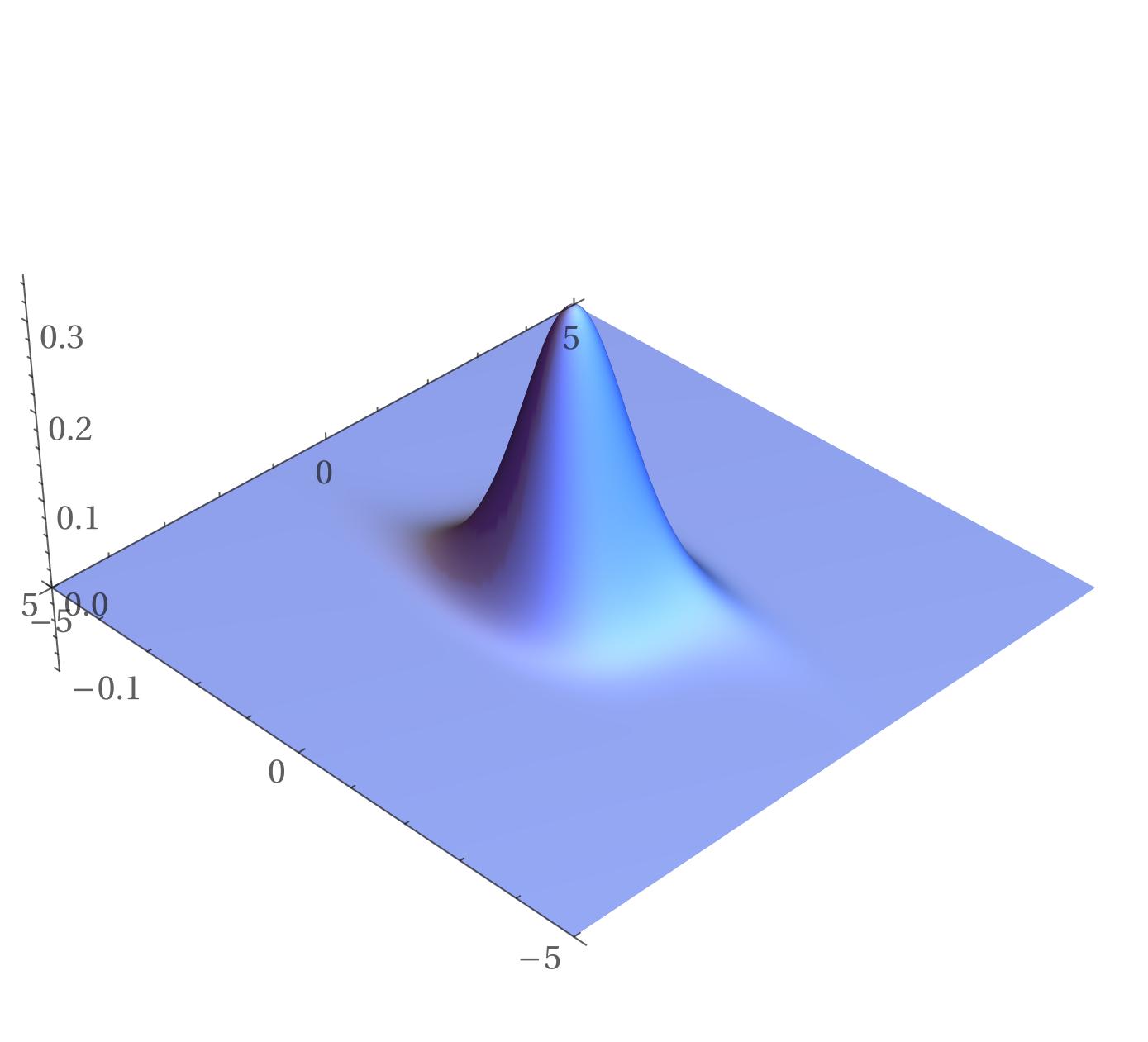}
\hspace{2em}
\includegraphics[width=.4\linewidth]{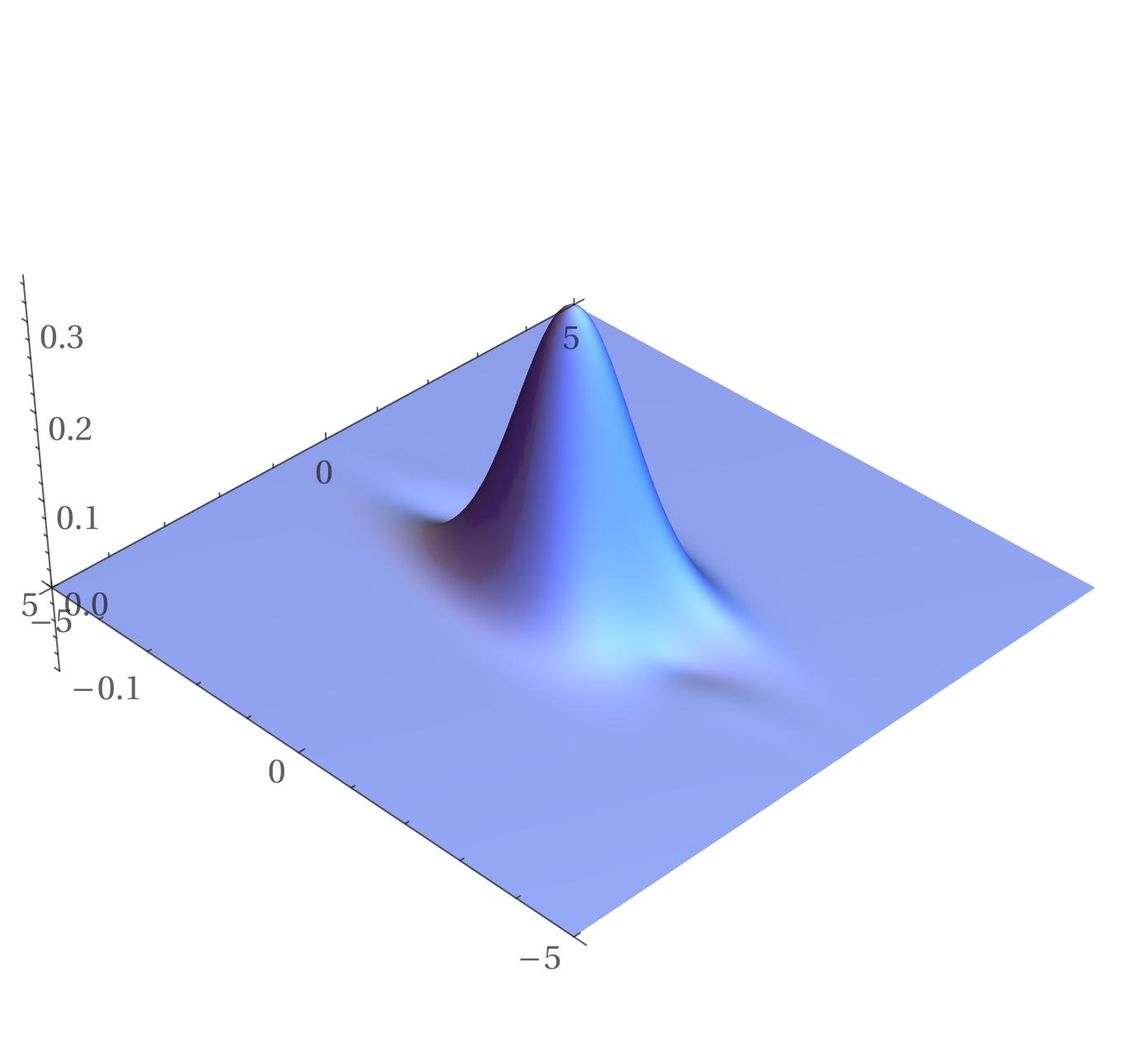}
\caption{Time-evolved state for a time-dependent anharmonic potential in Eq.~\eqref{eq:timedep-anharmonic}. The left panel is for the classical time evolution using the Poisson bracket, and the right panel is for the quantum time evolution using the Moyal bracket. In either case, the initial state is the Wigner function of the ground state of the quantum harmonic oscillator and the time-dependence $\lambda(t)$ of the anharmonic term is chosen in a wedge-like form. The left axes is $p$ in units of $\sqrt{m\omega\hbar}$ and right axis is $q$ in units of $\sqrt{\hbar/m\omega}$. The height is in units of $1/\hbar$.}
\label{fig:time-state}
\end{figure}

For a quantitative evaluation, we assume that $\lambda(t)$ has a wedge-like shape $\Lambda(\tau)=\max(0,1-\abs{1-2 \tau})$, more specifically, we choose $\lambda(t)=\frac13\Lambda(\frac{4\omega}\pi t)$. We assume that the initial state is the ground state of the quantum harmonic oscillator, $\rho_0 =\frac1{\pi\hbar}\exp(-H(0)/\frac{\hbar\omega}2)$ and due to its Gaussian shape, it is numerically more stable to solve the time-evolution for the state than for $p^2$. We use a numerical solver for the corresponding partial differential equation and obtain the time-evolved state at $t=\frac\pi{4\omega}$ as displayed in Figure~\ref{fig:time-state}, both for the classical ($a_1=0$) and quantum ($a_1=-\frac1{24})$ case. Since the mean value of the momentum vanishes, $\braket{\tilde p}=0$, for symmetry reasons and since we have $\tilde p^2=\widetilde{p^2}$, the variance of the momentum is simply given by $(\Delta\tilde p)^2=\braket{\widetilde{p^2}}$. Evaluating the corresponding phase-space integral on our numerical evolution, we find
\begin{equation}
\frac{(\Delta\tilde p)^2}{\hbar\omega m} \approx 0.6795+0.0823\, a_1
\end{equation}
for the variance at $t=\frac\pi{4\omega}$.

\section{Properties of the generalized Moyal bracket} \label{appendix:propertiesGMB}
In this section we prove several properties of the generalized Moyal bracket, defined as
\begin{equation}
\mb{f,g} = \sum_{n=0}^\infty a_n \hbar^{2n} f D_\omega^{2n+1} g,
\end{equation}
where $a_0=1$ and
\begin{equation}
D_\omega = \sum_{i=1}^N \dfrac{\overleftarrow{\partial}}{\partial q_i} \dfrac{\overrightarrow{\partial}}{\partial p_i} - \dfrac{\overleftarrow{\partial}}{\partial p_i} \dfrac{\overrightarrow{\partial}}{\partial q_i}.
\end{equation}

\begin{proposition}
The generalized Moyal bracket is linear, i.e., we have
\begin{align}
\mb{\alpha_1 f_1 + \alpha_2 f_2, g} & = \alpha_1 \mb{f_1, g} + \alpha_2 \mb{f_2, g}, \\
\mb{f, \beta_1 g_1 + \beta_2 g_2} & = \beta_1 \mb{f, g_1} + \beta_2 \mb{f, g_2}.
\end{align}
\end{proposition}
\begin{proof}
The proof follows by induction. Assume that
\begin{equation}
(f_1 + \alpha f_2) D_\omega^n g = f_1 D_\omega^n g + \alpha f_2 D_\omega^n g,
\end{equation}
then
\begin{align}
(f_1 + \alpha f_2) D_\omega^{n+1} g & = \sum_{i=1}^N \left( \dfrac{\partial (f_1 + \alpha f_2)}{\partial q_i} D_\omega^{n} \dfrac{\partial g}{\partial p_i} - \dfrac{\partial (f_1 + \alpha f_2)}{\partial p_i} D_\omega^{n} \dfrac{\partial g}{\partial q_i} \right) \\ &= \sum_{i=1}^N \left( \dfrac{\partial f_1}{\partial q_i} D_\omega^{n} \dfrac{\partial g}{\partial p_i} - \dfrac{\partial f_1}{\partial p_i} D_\omega^{n} \dfrac{\partial g}{\partial q_i} \right) + \alpha \sum_{i=1}^N \left( \dfrac{\partial f_2}{\partial q_i} D_\omega^{n} \dfrac{\partial g}{\partial p_i} - \dfrac{\partial f_2}{\partial p_i} D_\omega^{n} \dfrac{\partial g}{\partial q_i} \right) \\ &= f_1 D_\omega^{n+1} g + \alpha_2 f_2 D_\omega^{n+1} g.
\end{align}
Linearity of the generalized Moyal bracket follows from its definition. Proof of linearity in the second argument is analogical.
\end{proof}

In order to show that $\mb{\cdot,\cdot}$ is anti-symmetric we will need the following result.
\begin{proposition}
\label{prop:propertiesGMB-symAntisym}
Let $k \in \{0, 1, 2, \ldots \}$, then
\begin{align}
f D_\omega^{2k} g & = g D_\omega^{2k} f, \\
f D_\omega^{2k+1} g & = - g D_\omega^{2k+1} f.
\end{align}
\end{proposition}
\begin{proof}
The proof follows by induction: we will prove that if $f D_\omega^{2k} g = g D_\omega^{2k} f$ then $f D_\omega^{2k+1} g = -g D_\omega^{2k+1} f$ and that if $f D_\omega^{2k-1} g = -g D_\omega^{2k-1} f$, then $f D_\omega^{2k} g = g D_\omega^{2k} f$. The result then follows from $f D_\omega g = -g D_\omega f$ by alternating between the induction steps.

So assume that $f D_\omega^{2k} g = g D_\omega^{2k} f$, then we have
\begin{align}
f D_\omega^{2k+1} g & = \sum_{i=1}^N \left( \dfrac{\partial f}{\partial q_i} D_\omega^{2k} \dfrac{\partial g}{\partial p_i} - \dfrac{\partial f}{\partial p_i} D_\omega^{2k} \dfrac{\partial g}{\partial q_i} \right) \\ &= \sum_{i=1}^N \left( \dfrac{\partial g}{\partial p_i} D_\omega^{2k} \dfrac{\partial f}{\partial q_i} - \dfrac{\partial g}{\partial q_i} D_\omega^{2k} \dfrac{\partial f}{\partial p_i} \right) \\ &= - g D_\omega^{2k+1} f.
\end{align}
Assuming that $f D_\omega^{2k-1} g = -g D_\omega^{2k-1} f$, we get
\begin{align}
f D_\omega^{2k} g & = \sum_{i=1}^N \left( \dfrac{\partial f}{\partial q_i} D_\omega^{2k-1} \dfrac{\partial g}{\partial p_i} - \dfrac{\partial f}{\partial p_i} D_\omega^{2k-1} \dfrac{\partial g}{\partial q_i} \right) \\ &= \sum_{i=1}^N \left( - \dfrac{\partial g}{\partial p_i} D_\omega^{2k-1} \dfrac{\partial f}{\partial q_i} + \dfrac{\partial g}{\partial q_i} D_\omega^{2k-1} \dfrac{\partial f}{\partial p_i} \right) \\ &= g D_\omega^{2k} f.
\end{align}
\end{proof}

\begin{corollary}
$\mb{f,g} = -\mb{g,f}$
\end{corollary}
\begin{proof}
The result follows since $\mb{\cdot,\cdot}$ contains only odd powers of $D_\omega$.
\end{proof}

\begin{proposition}
Let $P_2$ be a polynomial of at most second order, then $\mb{f,P_2} = \{f, P_2\}$.
\end{proposition}
\begin{proof}
$P_2$ is a polynomial of at most second order if it is of the form
\begin{equation}
P_2(\vec{q},\vec{p}) = \sum_{i,j=1}^N ( J_{ij} q_i q_j + K_{ij} p_i p_j + L_{ij} q_i p_j) + \sum_{i=1}^N (a_i q_i + b_i p_i) + c.
\end{equation}
Clearly any third derivative of $P_2$ is zero, hence we have
\begin{equation}
f D_\omega^{2k+1} P_2 = 0
\end{equation}
for $k \geq 1$. It thus follows that $\mb{f,P_2} = \{f, P_2\}$.
\end{proof}

\begin{proposition}
Let $f,g,h$ be functions on the phase space such that their product and products of their derivatives vanish at infinity. Then
\begin{equation}
\label{eq:propertiesGMB-DomegaPerPartes}
\int_{\RR^{2N}} f \left( g D^k_\omega h \right) \ddd{N} q \ddd{N} p = \int_{\RR^{2N}} \left( f D^k_\omega g \right) h \ddd{N} q \ddd{N} p
\end{equation}
\end{proposition}
\begin{proof}
The result follows by induction. We have:
\begin{align}
\int_{\RR^{2N}} f (g D^{k+1}_\omega h) \ddd{N} q \ddd{N} p & = \sum_{i=1}^N \left( \int_{\RR^{2N}} f \left( \dfrac{\partial g}{\partial q_i} D^k_\omega \dfrac{\partial h}{\partial p_i} \right) \ddd{N} q \ddd{N} p - \int_{\RR^{2N}} f \left( \dfrac{\partial g}{\partial p_i} D^k_\omega \dfrac{\partial h}{\partial q_i} \right) \ddd{N} q \ddd{N} p \right) \\ &= - \sum_{i=1}^N \left( \int_{\RR^{2N}} \dfrac{\partial f}{\partial p_i} \left( \dfrac{\partial g}{\partial q_i} D^k_\omega h \right) \ddd{N} q \ddd{N} p - \int_{\RR^{2N}} \dfrac{\partial f}{\partial q_i} \left( \dfrac{\partial g}{\partial p_i} D^k_\omega h \right) \ddd{N} q \ddd{N} p \right) \\ &= - \sum_{i=1}^N \left( \int_{\RR^{2N}} \left( \dfrac{\partial f}{\partial p_i} D^k_\omega \dfrac{\partial g}{\partial q_i} \right) h \ddd{N} q \ddd{N} p - \int_{\RR^{2N}} \left( \dfrac{\partial f}{\partial q_i} D^k_\omega \dfrac{\partial g}{\partial p_i} \right) h \ddd{N} q \ddd{N} p \right) \\ &= \int_{\RR^{2N}} \left( f D^{k+1}_\omega g \right) h \ddd{N} q \ddd{N} p
\end{align}
where in the first step we have used integration by parts, the terms containing second derivatives of $g$ cancel each other. In the second step we have used the induction assumption.

One may conclude the proof by either checking that \eqref{eq:propertiesGMB-DomegaPerPartes} holds for $k=1$, or one may define $f D^0_\omega g = fg$, then it is trivial to check that \eqref{eq:propertiesGMB-DomegaPerPartes} holds for $k=0$.
\end{proof}

\begin{corollary}
Let $f,g,h$ be functions on the phase space such that their product and products of their derivatives vanish at infinity. Then
\begin{equation}
\label{eq:propertiesGMB-LhPerPartes}
\int_{\RR^{2N}} f (\Lie_h g) \ddd{N} q \ddd{N} p = - \int_{\RR^{2N}} (\Lie_h f) g \ddd{N} q \ddd{N} p.
\end{equation}
\end{corollary}
\begin{proof}
The result follows by expressing the generalized Moyal bracket in powers of the operator $D_\omega$ and using the the antisymmetry of the bracket.
\end{proof}

We will use the following result in order to argue that only odd powers of $D_\omega$ can be included in the definition of $\mb{\cdot,\cdot}$.
\begin{proposition} \label{prop:propertiesGMB-evenPolynomials}
$(qp)^n D^{2n}_\omega (qp)^n = (-1)^n (2n)! (n!)^2$.
\end{proposition}
\begin{proof}
All terms in the expansion of $D^{2n}_\omega$ will contain $2n$ derivatives and we get non-zero contribution only from terms with the same number of position and momentum derivatives. Counting the number of such terms is the same as counting the number of orderings of $n$ identical black and $n$ identical white balls, since we are essentially counting the number of ways in which we obtain the correct derivative by adding either derivatives with respect to position or momentum. Thus one can see that there is exactly $\frac{(2n)!}{(n!)^2}$ of these factors. Moreover all of these terms contain the factor $(-1)^n$ coming from the momentum derivative. We thus get:
\begin{equation}
(qp)^n D^{2n}_\omega (qp)^n = (-1)^n \dfrac{(2n)!}{(n!)^2} \left( \dfrac{\partial^{2n} (qp)^n}{\partial^n q \partial^n p} \right)^2 = (-1)^n (2n)! (n!)^2.
\end{equation}
\end{proof}

\begin{corollary}
Assume that $\mb{\cdot,\cdot}$ includes even powers of $D_\omega$, then there is a function $f$ such that $\mb{f,f} \neq 0$.
\end{corollary}
\begin{proof}
Let $\mb{\cdot,\cdot}$ be defined as:
\begin{equation}
\mb{f,g} = \sum_{k=0}^\infty b_k \hbar^{k-1} f D_\omega^{k} g
\end{equation}
and let $n \in \NN$ be the smallest number such that $b_{2n} \neq 0$. We then have
\begin{equation}
\mb{(qp)^n, (qp)^n} = b_{2n} \hbar^{2n-1} (qp)^n D^{2n}_\omega (qp)^n,
\end{equation}
because $(qp)^n D^{2 \ell + 1}_\omega (qp)^n = 0$ due to Proposition~\ref{prop:propertiesGMB-symAntisym} and $(qp)^n D^{2 \ell}_\omega (qp)^n = 0$ for $\ell > n$ because then the order of derivatives in $D^{2 \ell}_\omega$ is strictly higher than $n$. The result follows by Proposition~\ref{prop:propertiesGMB-evenPolynomials}.
\end{proof}

\begin{corollary}
Assume that $\mb{\cdot,\cdot}$ includes even powers of $D_\omega$, then either the resulting time-evolution is not Markovian, or the generator of time-translations and energy observable coincide only for $t=0$.
\end{corollary}
\begin{proof}
Assume that $\mb{\cdot,\cdot}$ includes even powers of $D_\omega$ and let $H$ be a function such that $\mb{H,H} \neq 0$. But then in Heisenberg picture we have $\dot{H} \neq 0$ and so $H = H(t)$ depends on time. Let $\mathcal{G}$ be the generator of the time-evolution constructed from the Hamiltonian using the generalized Moyal bracket. We now have two options: we either allow the generator to evolve itself in time and we have $\mathcal{G} = \Lie_{H(t)}$, or we keep the generator constant in time and we have $\mathcal{G} = \Lie_H = \Lie_{H(0)}$. In the first case the time-evolution is not Markovian anymore since the generator depends on time, in the second case the energy observable $H(t)$ corresponds to the generator $\mathcal{G}$ only for $t=0$.
\end{proof}

\begin{example}[Example of functions such that $\mb{f,g} = 0$ but $\{f,g\} \neq 0$]
For the Pöschl--Teller potential we have the Hamiltonian
\begin{equation}
H(q,p) = \frac{p^2}{2m} +
\frac{\eta^2}{2m}\left(1 - \dfrac{2}{\cosh^2(q\eta / \hbar)}\right)
\end{equation}
where $\eta$ is some constant with units of momentum. The ground state has energy $0$ and the Wigner function \cite{CurtrightFairlieZachos-timeIndependentWignerFunction}
\begin{equation}
\rho(q,p) = \frac1\hbar\dfrac{\sin(2qp / \hbar)}{\sinh(2q\eta / \hbar) \sinh(\pi p/\eta)}.
\end{equation}
So we must have $\mb{H, \rho} = 0$ for the Moyal bracket used in quantum theory, but it is straightforward to check that $\{H,\rho\} \neq 0$. We mention that the occurrence of $\hbar$ in the Hamiltonian is essential and it cannot be easily replaced by another constant with the same units.
\end{example}

\begin{example}[Example of functions such that $\{f,g\} = 0$ but $\mb{f,g} \neq 0$]
Observe that for any function $f: \RR^2 \to \RR$ we have $\{f, f^2\} = 0$; we will show that analogical result does not hold for the generalized Moyal bracket. One can also show that we have
\begin{equation}
(q^n p^{n+1}) D_\omega^{2n+1} (q^{2n} p^{2n+2}) = (-1)^{n+1} \dfrac{(2n)! (2n+1)! (2n+2)!}{(n-1)! (n+2)!} q^{n-1} p^{n+2}.
\end{equation}
Let then $n \geq 1$ be the lowest index such that $a_n \neq 0$, where $a_n$ are the coefficients used in the definition of the generalized Moyal bracket in Eq.~\eqref{eq:time-generalizedMoyal}. Consequently, we have $\mb{f,f^2}\ne 0$ for $f=q^np^{n+1}$.
\end{example}

\begin{proposition} \label{prop:propertiesGMB-max2Energy}
Let $H:\RR^{2N} \to \RR$ be a Hamiltonian and let $g_H(I)$ be the corresponding spectral measure and assume that the energy spectrum is discrete, i.e., that there are energies $E_n \in \RR$, $n \in \NN$, such that $g_H(I;t) = \sum_{n: E_n \in I}g_H(E_n;t)$. Let $\rho$ be a state such that $\Pr[\tilde{H} = E_k] \neq 0$ for at most two $k$, that is $\Pr[\tilde{H} = E_k] \neq 0$ only for $k \in \{n_1, n_2\}$. Then $\frac{d}{dt} \Pr[\tilde{H} = E_k] = 0$.
\end{proposition}
\begin{proof}
From the assumptions it follows that at $t=0$,
\begin{align}
1 & = \braket{g_H(E_{n_1}), \rho} + \braket{g_H(E_{n_2}), \rho}, \\
\braket{\tilde{H}} & = E_{n_1} \braket{g_H(E_{n_1}), \rho} + E_{n_2} \braket{g_H(E_{n_2}), \rho}.
\end{align}
Taking the time derivative in Heisenberg picture at $t=0$ we get
\begin{align}
0 & = \braket{\dot{g}_H(E_{n_1}), \rho} + \braket{\dot{g}_H(E_{n_2}), \rho}, \\
0 & = E_{n_1} \braket{\dot{g}_H(E_{n_1}), \rho} + E_{n_2} \braket{\dot{g}_H(E_{n_2}), \rho}
\end{align}
and since $E_{n_1} \neq E_{n_2}$ we get $\braket{\dot{g}_H(E_{n_1}), \rho}=0$ and $\braket{\dot{g}_H(E_{n_2}), \rho}=0$.
\end{proof}

\section{Phase space spectral measures of energy and angular momentum of the hydrogen atom} \label{appendix:hydrogen}
The spectrum of energies of bound states is discrete and we set $E_n = -\frac{\kappa}{2 a_0} \frac{1}{n^2}$. The phase space spectral measure of energy, $g_H(I)$, is given as
\begin{equation}
g_H(I; \vec{q}, \vec{p}) = \sum_{n: E_n \in I} T_n(H(\vec{q}, \vec{p})
\end{equation}
where $T_n^H$ are the piecewise linear functions. If we assume that $g_H(E_n; \vec{q}, \vec{p}) \neq 0$ only if $H(\vec{q}, \vec{p}) \in [E_{n-1}, E_{n+1}]$, then we get $T_n(x) \neq 0$ only if $x \in [E_{n-1}, E_{n+1}]$. For $x \in [E_n, E_{n+1}]$ the normalization and expectation value conditions on $g_H(I; \vec{q}, \vec{p})$ become
\begin{align}
T_n(x) + T_{n+1}(x) & = 1, \\
E_n T_n(x) + E_{n+1} T_{n+1}(x) & = x. \\
\end{align}
This series of linear equations has unique solution given by the sawtooth functions:
\begin{align}
T_1^H (x) & =
\begin{cases}
\frac{E_2 - x}{E_2 - E_1} & x \leq E_2 \\
0 & x \geq E_2
\end{cases}
\\
T_2^H (x) & =
\begin{cases}
\frac{x - E_1}{E_2 - E_1} & x \leq E_2 \\
\frac{E_3 - x}{E_3 - E_2} & x \in [E_2,E_3] \\
0 & x \geq E_3
\end{cases}
\end{align}
and for $n \geq 3$ as
\begin{equation}
T_n^H (x) =
\begin{cases}
\frac{x - E_{n-1}}{E_n - E_{n-1}} & x \in [E_{n-1},E_n] \\
\frac{E_{n+1} - x}{E_{n+1} - E_n} & x \in [E_n,E_{n+1}] \\
0 & x \notin [E_{n-1},E_{n+1}]
\end{cases}
\end{equation}
Note that for $n \neq 2$ we have $T_n^H(x) \geq 0$ for all $x$, while $T_2^H(x) \leq 0$ for $x \leq E_1$ but $T_2^H(x) \geq 0$ for $x \leq E_2$.

Since this is only the spectral measure for bound states, we only need to check the conditions that the phase space spectral measure has satisfy for negative energies, i.e., only for $I \subset \RR_-$, where $\RR_-$ is the set of negative real numbers. The normalization condition $g_H(\RR_-) = 1$, follows from $\sum_{n=1}^\infty T_n^H(x) = 1$, while the expectation value condition $\sum_{n=1}^\infty g_H(E_n; \vec{q}, \vec{p}) = H(\vec{q}, \vec{p})$ follows from $\sum_{n=1}^\infty E_n T_n^H(x) = x$.

The phase space spectral measure for angular momentum is constructed in a similar way,
\begin{equation}
g_{L_i}(I; \vec{q}, \vec{p}) = \sum_{m: m \hbar \in I} T_m^{L_i}(L_i(\vec{q}, \vec{p}))
\end{equation}
for $m \in \ZZ$. Here $T_m^{L_i}$ are the sawtooth functions given as
\begin{equation}
T_m^{L_i} (x) =
\begin{cases}
\frac{x - (m-1)\hbar}{\hbar} & x \in [(m-1)\hbar, m\hbar] \\
\frac{(m+1)\hbar - x}{\hbar} & x \in [m\hbar, (m+1)\hbar] \\
0 & x \notin [(m-1)\hbar, (m+1)\hbar]
\end{cases}
\end{equation}
The normalization condition $g_{L_i}(\RR) = 1$ follows from $\sum_{m \in \ZZ} T_m^{L_i} = 1$ and the expectation value condition $\sum_{m \in \ZZ} m \hbar g_{L_i}(m \hbar; \vec{q}, \vec{p}) = L_i(\vec{q}, \vec{p})$ follows from $\sum_{m \in \ZZ} m \hbar T_m^{L_i}(x) = x$.

\begin{proposition}
We have $g_H(E_n) g_{L_i}(m \hbar) \neq 0$ only if $\abs{m} \leq 2(n+1)$.
\end{proposition}
\begin{proof}
Let $n \in \NN$ and $m \in \ZZ$ be such that $g_H(E_n) g_{L_i}(m \hbar) \neq 0$, which implies $g_H(E_n) \neq 0$ and $g_{L_i}(m \hbar) \neq 0$. We have $g_H(E_n) \neq 0$ only if $H(\vec{q}, \vec{p}) \in (E_{n-1}, E_{n+1})$, this implies
\begin{equation} \label{eq:hydrogen-nmRel-EnBound}
- \dfrac{\abs{\vec{p}}^2}{2 \mu} + \dfrac{\kappa}{\abs{\vec{q}}} > \abs{E_{n+1}}.
\end{equation}
We have
\begin{equation}
\abs{E_{n+1}} < - \dfrac{\abs{\vec{p}}^2}{2 \mu} + \dfrac{\kappa}{\abs{\vec{q}}} \leq \dfrac{\kappa}{\abs{\vec{q}}}
\end{equation}
which yields
\begin{equation} \label{eq:hydrogen-nmRel-qBound}
\abs{\vec{q}} < \dfrac{\kappa}{\abs{E_{n+1}}}.
\end{equation}
Using Eq.~\eqref{eq:hydrogen-nmRel-EnBound} we also get
\begin{equation}
\abs{E_{n+1}} + \dfrac{\abs{\vec{p}}^2}{2 \mu} < \dfrac{\kappa}{\abs{\vec{q}}}
\end{equation}
from which we get
\begin{equation} \label{eq:hydrogen-nmRel-p2qBound}
\abs{\vec{p}}^2 \abs{\vec{q}} < 2\mu ( \kappa - \abs{E_{n+1}} \abs{\vec{q}} ) \leq 2\mu \kappa.
\end{equation}
We have $g_{L_i}(m \hbar) \neq 0$ only if $L_i(\vec{q}, \vec{p}) \in ((m-1)\hbar, (m+1)\hbar)$ which implies
\begin{equation}
(\abs{m} - 1) \hbar < \abs{L_i}.
\end{equation}
Using the formula for the norm of cross product we get $\abs{L_i} \leq \abs{\vec{L}} \leq \abs{\vec{q}} \abs{\vec{p}}$ and we get
\begin{equation}
(\abs{m} - 1) \hbar < \abs{\vec{q}} \abs{\vec{p}}.
\end{equation}
By squaring both sides and using Eq.~\eqref{eq:hydrogen-nmRel-p2qBound} we get
\begin{equation}
(\abs{m} - 1)^2 \hbar^2 < \abs{\vec{q}}^2 \abs{\vec{p}}^2 < 2 \mu \kappa \abs{\vec{q}}
\end{equation}
and finally using Eq.~\eqref{eq:hydrogen-nmRel-qBound} we obtain
\begin{equation}
(\abs{m} - 1)^2 \hbar^2 < 2 \mu \kappa \abs{\vec{q}} < 2 \mu \dfrac{\kappa^2}{\abs{E_{n+1}}} = 4 \mu \kappa a_0 (n+1)^2 = 4 \hbar^2 (n+1)^2.
\end{equation}
Taking a square root we obtain the final expression
\begin{equation}
\abs{m} \leq 2(n+1).
\end{equation}
\end{proof}

\section{Hydrogen atom in a non-stationary electric field} \label{appendix:electric}
Let
\begin{align}
H(\vec{q}, \vec{p}) & = \dfrac{\abs{\vec{p}}^2}{2 \mu} - \dfrac{\kappa}{\abs{\vec{q}}}, \\
H_e(\vec{q}, \vec{p}, t) & = -2 eE \sin(\omega t) q_3,
\end{align}
be the corresponding Hamiltonians. Let $f: \RR^6 \to \RR$ be a function on the phase space, then we have
\begin{equation}
\Lie_{H_e} f(\vec{q}. \vec{p}) = -2 eE \sin(\omega t) \dfrac{\partial}{\partial p_3} f(\vec{q}, \vec{p})
\end{equation}
and
\begin{equation}
\Lie_H f(\vec{q}, \vec{p}) = - \dfrac{\vec{p}}{\mu} \cdot \vec{\nabla}_q f(\vec{q}, \vec{p}) - \mb{\dfrac{\kappa}{\abs{\vec{q}}}, f(\vec{q}, \vec{p})}.
\end{equation}
We then have
\begin{align}
\Lie_{H_e} \Lie_H f(\vec{q}, \vec{p}) & = -2 eE \sin(\omega t) \dfrac{\partial}{\partial p_3} \left( - \dfrac{\vec{p}}{\mu} \cdot \vec{\nabla}_q f(\vec{q}, \vec{p}) - \mb{\dfrac{\kappa}{\abs{\vec{q}}}, f(\vec{q}, \vec{p})} \right) \\ &= -2 eE \sin(\omega t) \left( - \dfrac{1}{\mu} \dfrac{\partial}{\partial q_3} f(\vec{q}, \vec{p}) - \dfrac{\vec{p}}{\mu} \cdot \vec{\nabla}_q \dfrac{\partial}{\partial p_3} f(\vec{q}, \vec{p}) - \mb{\dfrac{\kappa}{\abs{\vec{q}}}, \dfrac{\partial}{\partial p_3} f(\vec{q}, \vec{p})} \right) \\ &= 2 eE \sin(\omega t) \dfrac{1}{\mu} \dfrac{\partial}{\partial q_3} f(\vec{q}, \vec{p}) + \Lie_H \Lie_{H_e} f(\vec{q}, \vec{p})
\end{align}
and we get
\begin{equation}
\Lie_{H_e} \Lie_H - \Lie_H \Lie_{H_e} = 2 eE \sin(\omega t) \dfrac{1}{\mu}
\dfrac{\partial}{\partial q_3} = \Lie_{G}
\end{equation}
for
\begin{equation}
G = - \dfrac{2 eE}{\mu} \sin(\omega t) p_3.
\end{equation}

\section{Scattering theory of the hydrogen atom} \label{appendix:scattering}
Let $H(\vec{q}, \vec{p}) = \frac{\abs{\vec{p}}^2}{2 \mu} + V(\vec{q})$ be a Hamiltonian, we will decompose the time-evolution equation for a state $\rho$ as follows:
\begin{align}
\dot{\rho} = \mb{H, \rho} & = \{\frac{\abs{\vec{p}}^2}{2 \mu}, \rho(t; \vec{q}, \vec{p})\} + \mb{V(\vec{q}), \rho(t; \vec{q}, \vec{p})} \\
 & = \{\frac{\abs{\vec{p}}^2}{2 \mu}, \rho(t; \vec{q}, \vec{p}) \} + \int\limits_{\RR^3} \mb{V(\vec{q}), \delta^{(3)}(\vec{p} - \vec{p'})} \rho(t; \vec{q}, \vec{p'}) \ddd{3} p',
\end{align}
where we have used that in $\mb{V(\vec{q}), \rho(t; \vec{q}, \vec{p})}$ only partial derivatives in $\vec{p}$ act on $\rho(t; \vec{q}, \vec{p})$. Let us denote
\begin{equation}
K(\vec{q}, \vec{p}, \vec{p'}) = \mb{V(\vec{q}), \delta^{(3)}(\vec{p} - \vec{p'})}.
\end{equation}
Also note that
\begin{equation}
\{\frac{\abs{\vec{p}}^2}{2 \mu}, \rho(t; \vec{q}, \vec{p})\} = - \frac{\vec{p}}{\mu} \cdot \vec{\nabla}_q \rho(t; \vec{q}, \vec{p})
\end{equation}
where $\vec{\nabla}_q$ is the gradient in $\vec{q}$. Putting all this together we get a new form of the time-evolution equation:
\begin{equation}
\label{eq:scattering-timeEvolution}
\dot{\rho}(t; \vec{q}, \vec{p}) + \frac{\vec{p}}{\mu} \cdot \vec{\nabla}_q \rho(t; \vec{q}, \vec{p}) = \int\limits_{\RR^3} K(\vec{q}, \vec{p}, \vec{p'}) \rho(t; \vec{q}, \vec{p'}) \ddd{3} p'.
\end{equation}
In order to solve this equation we will use the fact that the retarded Green's function for the differential operator on left-hand-side is known. That is, let
\begin{equation}
\label{eq:scattering-greens}
G(\vec{q}, \vec{p}, t) = \theta(t) \delta^{(3)} (\vec{q} - \frac{\vec{p} t}{\mu}),
\end{equation}
where $\theta(t)$ is the Heaviside step function and $\delta^{(3)} (\vec{q} - \frac{\vec{p} t}{\mu})$ is to be integrated over $\vec{q}$. Then we have
\begin{equation}
\dot{G}(\vec{q}, \vec{p}, t) + \frac{\vec{p}}{\mu} \cdot \vec{\nabla}_q G(\vec{q}, \vec{p}, t) = \delta(t) \delta^{(3)}(\vec{q}).
\end{equation}

Before we construct the solution of the time-evolution that we aim for, we have to discuss initial conditions. In this case we want to have initial conditions formally at $t = -\infty$. In order to define these initial conditions in a meaningful way let $\rho_{\ini}$ be the solution of the free time-evolution equation, that is,
\begin{equation}
\label{eq:scattering-indyn}
\dot{\rho}_{\ini} = \{ \dfrac{\abs{\vec{p}}^2}{2 \mu}, \rho_{\ini} \}.
\end{equation}
Then as an initial condition we require that (pointwise),
\begin{equation} \label{eq:scattering-initial}
\lim_{t \to -\infty}(\rho - \rho_{\ini}) = 0.
\end{equation}
One finds easily that
\begin{equation} \label{eq:scattering-sol}
\rho(t; \vec{q}, \vec{p}) = \rho_{\ini}(t; \vec{q}, \vec{p}) + \int\limits_{\RR} \int\limits_{\RR^3} \int\limits_{\RR^3} G(\vec{q} - \vec{q'}, \vec{p}, t - \tau) K(\vec{q'}, \vec{p}, \vec{p'}) \rho(\tau; \vec{q'}, \vec{p'}) \ddd{3} p' \ddd{3} q' \dd \tau.
\end{equation}
satisfies the initial condition \eqref{eq:scattering-initial}. In addition, applying Eq.~\eqref{eq:scattering-sol} to the left hand side of Eq.~\eqref{eq:scattering-timeEvolution} gives immediately the right hand side of Eq.~\eqref{eq:scattering-timeEvolution}, by virtue of Eq.~\eqref{eq:scattering-greens} and Eq.~\eqref{eq:scattering-indyn}.

Eq.~\eqref{eq:scattering-sol} is still an integral equation but we will be able to solve it iteratively. Plugging in the expression for $G$ and integrating we get:
\begin{align}
\rho(t; \vec{q}, \vec{p}) & = \rho_{\ini}(t; \vec{q}, \vec{p}) + \int\limits_{\RR} \int\limits_{\RR^3} \int\limits_{\RR^3} \theta(t - \tau) \delta(\vec{q} - \vec{q'} - \frac{\vec{p}(t - \tau)}{\mu}) K(\vec{q'}, \vec{p}, \vec{p'}) \rho(\tau; \vec{q'}, \vec{p'}) \ddd{3} p' \ddd{3} q' \dd \tau \\
 & = \rho_{\ini}(t; \vec{q}, \vec{p}) + \int\limits_{\RR^3} \int\limits_{-\infty}^{t} K(\vec{q} - \frac{\vec{p}}{\mu}(t - \tau), \vec{p}, \vec{p'}) \rho(\tau; \vec{q} - \frac{\vec{p}}{\mu}(t - \tau), \vec{p'}) \dd \tau \ddd{3} p'. \label{eq:scattering-EQ}
\end{align}

We can solve the equation above in a perturbative manner via $V(\vec{q}) \to \lambda V(\vec{q})$ and the expansion
\begin{equation}
\rho(t;\vec{q},\vec{p}) = \sum_{n=0}^\infty \sum_{k=0}^\infty \hbar^{2n} \lambda^k \rho_{n,k}(t; \vec{q},\vec{p}).
\end{equation}
Our strategy is now to determine how $\rho_{n,k}(t; \vec{q},\vec{p})$ depends on $\abs{\vec{q}}$. We do this since we assume that the detector is positioned far from the center of the potential, so we are interested only in the limit $\abs{\vec{q}} \to \infty$. In order to simplify the notation let $a_0 = 1$, then after plugging this expression into Eq.~\eqref{eq:scattering-EQ} we get
\begin{equation}
\begin{split}
\sum_{n=0}^\infty \sum_{k=0}^\infty \hbar^{2n} \lambda^k \rho_{n,k}(t; \vec{q},\vec{p}) &=
\rho_{\ini}(t; \vec{q}, \vec{p}) + \sum_{n=0}^\infty \sum_{n'=0}^\infty \sum_{k=0}^\infty a_{n'} \hbar^{2(n+n')} \lambda^{k+1} \\
&\int\limits_{\RR^3} \int\limits_{-\infty}^{t} \left( V(\vec{q} - \frac{\vec{p}}{\mu}(t - \tau)) D^{2n'+1}_\omega \delta^{(3)}(\vec{p} - \vec{p'}) \right) \rho_{n,k}(\tau; \vec{q} - \frac{\vec{p}}{\mu}(t - \tau), \vec{p'}) \dd \tau \ddd{3} p'.
\end{split}
\end{equation}
We proceed by comparing terms of the same powers in $\lambda$ on both sides. First of all observe that for $k = 0$ we have $\rho_{0,0} = \rho_{\ini}$ and $\rho_{n,0} = 0$ for $n > 0$. For $k \geq 1$ we get
\begin{equation}
\begin{split}
\sum_{n=0}^\infty \hbar^{2n} \rho_{n,k}(t; \vec{q},\vec{p}) = \sum_{n=0}^\infty \sum_{n'=0}^\infty a_{n'} \hbar^{2(n+n')} \int\limits_{\RR^3} \int\limits_{-\infty}^{t} &\left( V(\vec{q} - \frac{\vec{p}}{\mu}(t - \tau)) D^{2n'+1}_\omega \delta^{(3)}(\vec{p} - \vec{p'}) \right) \\
&\rho_{n,k-1}(\tau; \vec{q} - \frac{\vec{p}}{\mu}(t - \tau), \vec{p'}) \dd \tau \ddd{3} p'.
\end{split}
\end{equation}
Using the general identity $\sum_{n=0}^\infty \sum_{n'=0}^\infty f(n,n') =
\sum_{N=0}^\infty \sum_{\tilde{n}=0}^N f(\tilde{n}, N-\tilde{n})$ we get
\begin{equation}
\begin{split}
\sum_{n=0}^\infty \hbar^{2n} \rho_{n,k}(t; \vec{q},\vec{p}) = \sum_{N=0}^\infty \sum_{\tilde{n}=0}^N a_{(N-\tilde{n})} \hbar^{2 N} \int\limits_{\RR^3} \int\limits_{-\infty}^{t} &\left( V(\vec{q} - \frac{\vec{p}}{\mu}(t - \tau) ) D^{2(N-\tilde{n})+1}_\omega \delta^{(3)}(\vec{p} - \vec{p'}) \right) \\
&\rho_{\tilde{n},k-1}(\tau; \vec{q} - \frac{\vec{p}}{\mu}(t - \tau), \vec{p'}) \dd \tau \ddd{3} p'
\end{split}
\end{equation}
and by comparing the terms with the same power in $\hbar$ we obtain
\begin{equation}
\rho_{n,k}(t; \vec{q},\vec{p}) = \sum_{\tilde{n}=0}^n a_{(n-\tilde{n})} \int\limits_{\RR^3} \int\limits_{-\infty}^{t} \left( V(\vec{q} - \frac{\vec{p}}{\mu}(t - \tau) ) D^{2(n-\tilde{n})+1}_\omega \delta^{(3)}(\vec{p} - \vec{p'}) \right) \rho_{\tilde{n},k-1}(\tau; \vec{q} - \frac{\vec{p}}{\mu}(t - \tau), \vec{p'}) \dd \tau \ddd{3} p'.
\end{equation}
Moreover we perform the substitution $t - \tau = u$, we get
\begin{equation} \label{eq:scattering-rhonk}
\rho_{n,k}(t; \vec{q},\vec{p}) = \sum_{\tilde{n}=0}^n a_{(n-\tilde{n})} \int\limits_{\RR^3} \int\limits_{0}^{\infty} \left( V(\vec{q} - \frac{\vec{p}}{\mu} u) D^{2(n-\tilde{n})+1}_\omega \delta^{(3)}(\vec{p} - \vec{p'}) \right) \rho_{\tilde{n},k-1}(t-u; \vec{q} - \frac{\vec{p}}{\mu} u, \vec{p'}) \dd u \ddd{3} p'.
\end{equation}
We next prove that $\rho_{n,k}$ does not explicitly depend on time $t$, the physical intuition for this is that the initial state $\rho_{\ini}$ does not depend on time and since the experiment starts at $t = -\infty$, for any finite time the flow of the scattered particles must have stabilized. Mathematically we see from Eq.~\eqref{eq:scattering-rhonk} that $\rho_{n,k}$ depends on $t$ only if $\rho_{\tilde{n},k-1}$ depends on $t$ for some $\tilde{n} \in \{0, \ldots, n\}$. Since we already argued that $\rho_{\tilde{n}, 0}$ does not depend on $t$, we get that $\rho_{\tilde{n}, 1}$ does not depend on $t$. Proceeding by induction we get that $\rho_{n,k}$ does not depend on $t$ for all $k$ and $n$. For this reason we will drop the explicit time dependence and we will write
\begin{equation}
\rho_{n,k}(\vec{q}, \vec{p}) = \rho_{n,k}(t; \vec{q}, \vec{p})
\end{equation}
and
\begin{equation}
\rho(\vec{q}, \vec{p}) = \rho(t; \vec{q}, \vec{p}).
\end{equation}

Observe that in the expression $V(\vec{q} - \frac{\vec{p}}{\mu} u) D^{2(n-\tilde{n})+1}_\omega \delta^{(3)}(\vec{p} - \vec{p'})$ the nonzero contributions are only from terms where $\frac{\partial}{\partial q_i}$ acts on $V(\vec{q} - \frac{\vec{p}}{\mu} u)$ and so we have
\begin{equation}
V(\vec{q} - \frac{\vec{p}}{\mu} u) D^{n'}_\omega \delta^{(3)}(\vec{p} - \vec{p'}) = \sum_{\substack{i_1 = 1\\ \ldots \\ i_{n'} = 1}}^3 \left( \dfrac{\partial}{\partial q_{i_{n'}}} \cdots \dfrac{\partial}{\partial q_{i_1}} V(\vec{q} - \frac{\vec{p}}{\mu} u) \right) \left( \dfrac{\partial}{\partial p_{i_{n'}}} \cdots \dfrac{\partial}{\partial p_{i_1}} \delta^{(3)}(\vec{p} - \vec{p'}) \right),
\end{equation}
here $n' \in \NN$ is an odd number. Since we have
\begin{equation}
\dfrac{\partial}{\partial q_i} \dfrac{1}{\abs{\vec{q} - \dfrac{\vec{p}}{\mu} u}^k} = \dfrac{-k (q_i-\frac{p_i}{\mu} u)}{\abs{\vec{q} - \dfrac{\vec{p}}{\mu} u}^{k+2}}
\end{equation}
and every derivative either acts on the term in the denominator in this way, or acts on the polynomial in the numerator, we see that every derivative decreases the order $\abs{\vec{q}}$ by one. We thus have
\begin{equation}
V(\vec{q} - \frac{\vec{p}}{\mu} u) D^{2(n-\tilde{n})+1}_\omega \delta^{(3)}(\vec{p} - \vec{p'}) = \dfrac{\tilde{f}_{2(n-\tilde{n})+1}(\frac{\vec{q}}{\abs{\vec{q}}} - \frac{\vec{p}}{\mu} \frac{u}{\abs{\vec{q}}}, \vec{p}, \vec{p'}) }{\abs{\vec{q}}^{2(n - \tilde{n}) + 2}}
\end{equation}
where $\tilde{f}_{2(n-\tilde{n})+1}$ is some suitable function which also contains derivations of the Dirac distributions. Plugging this expression into Eq.~\eqref{eq:scattering-rhonk} we get
\begin{equation}
\rho_{n,k}(\vec{q},\vec{p}) = \sum_{\tilde{n}=0}^n a_{(n-\tilde{n})} \int\limits_{\RR^3} \int\limits_{0}^{\infty} \dfrac{\tilde{f}_{2(n-\tilde{n})+1}(\frac{\vec{q}}{\abs{\vec{q}}} - \frac{\vec{p}}{\mu} \frac{u}{\abs{\vec{q}}}, \vec{p}, \vec{p'}) }{\abs{\vec{q}}^{2(n - \tilde{n}) + 2}} \rho_{\tilde{n},k-1}(\vec{q} - \frac{\vec{p}}{\mu} u, \vec{p'}) \dd u \ddd{3} p'.
\end{equation}
Using the substitution $u = \abs{\vec{q}} v$ we get
\begin{equation} \label{eq:scattering-rhonkAbsq}
\rho_{n,k}(\vec{q},\vec{p}) = \sum_{\tilde{n}=0}^n \dfrac{a_{(n-\tilde{n})}}{\abs{\vec{q}}^{2(n - \tilde{n}) + 1}} \int\limits_{\RR^3} \int\limits_{0}^{\infty} \tilde{f}_{2(n-\tilde{n})+1}(\frac{\vec{q}}{\abs{\vec{q}}} - \frac{\vec{p}}{\mu} v, \vec{p}, \vec{p'}) \rho_{\tilde{n},k-1}(\vec{q} - \frac{\vec{p}}{\mu} \abs{\vec{q}} v, \vec{p'}) \dd v \ddd{3} p'.
\end{equation}
and we see that the last term inside the integral that depends on $\abs{\vec{q}}$ is $\rho_{\tilde{n},k-1}(\vec{q} - \frac{\vec{p}}{\mu} \abs{\vec{q}} v, \vec{p'})$. At this point it is useful to compute $\rho_{n,1}(\vec{q},\vec{p})$ to get some intuition for the following calculations. Using that $\rho_{0,0} = \rho_{\ini}$ and $\rho_{n,0} = 0$ for $n \geq 1$ we get
\begin{equation} \label{eq:scattering-rhon1Absq}
\rho_{n,1}(\vec{q},\vec{p}) = \dfrac{a_n}{\abs{\vec{q}}^{2n + 1}} \int\limits_{\RR^3} \int\limits_{0}^{\infty} \tilde{f}_{2n+1}(\frac{\vec{q}}{\abs{\vec{q}}}, \frac{\vec{q}}{\abs{\vec{q}}} - \frac{\vec{p}}{\mu} v, \vec{p}, \vec{p'}) \nu \delta^{(3)}(\vec{p'} - \vec{p}_0) \dd v \ddd{3} p'
\end{equation}
and so we see that the dependence on $\abs{\vec{q}}$ becomes explicit. We will proceed as follows: assume that for a given $k \in \NN$ and for all $n \in \NN$ we have
\begin{equation}
\rho_{n,k-1}(\vec{q}, \vec{p}) = \dfrac{\tilde{\rho}_{n,k-1}(\frac{\vec{q}}{\abs{\vec{q}}}, \vec{p})}{\abs{\vec{q}}^{\alpha(n,k-1)}}
\end{equation}
where $\alpha(n,k-1) \in \NN$. Then, using Eq.~\eqref{eq:scattering-rhonkAbsq} we get
\begin{equation}
\rho_{n,k}(\vec{q},\vec{p}) = \sum_{\tilde{n}=0}^n \dfrac{a_{(n-\tilde{n})}}{\abs{\vec{q}}^{2(n - \tilde{n}) + 1}} \int\limits_{\RR^3} \int\limits_{0}^{\infty} \tilde{f}_{2(n-\tilde{n})+1}(\frac{\vec{q}}{\abs{\vec{q}}} - \frac{\vec{p}}{\mu} v, \vec{p}, \vec{p'}) \dfrac{\tilde{\rho}_{\tilde{n},k-1} \left( \frac{\vec{q} - \frac{\vec{p}}{\mu} \abs{\vec{q}} v}{\abs{\vec{q} - \frac{\vec{p}}{\mu} \abs{\vec{q}} v}}, \vec{p'} \right)}{\abs{ \vec{q} - \frac{\vec{p}}{\mu} \abs{\vec{q}} v}^{\alpha(\tilde{n},k-1)}} \dd v \ddd{3} p'
\end{equation}
which leads to
\begin{equation}
\rho_{n,k}(\vec{q},\vec{p}) = \sum_{\tilde{n}=0}^n \dfrac{a_{(n-\tilde{n})}}{\abs{\vec{q}}^{2(n - \tilde{n}) + 1 + \alpha(\tilde{n},k-1)}} \int\limits_{\RR^3} \int\limits_{0}^{\infty} \tilde{f}_{2(n-\tilde{n})+1}(\frac{\vec{q}}{\abs{\vec{q}}} - \frac{\vec{p}}{\mu} v, \vec{p}, \vec{p'}) \dfrac{\tilde{\rho}_{\tilde{n},k-1} \left( \frac{ \frac{\vec{q}}{\abs{\vec{q}}} - \frac{\vec{p}}{\mu} v}{\abs{\frac{\vec{q}}{\abs{\vec{q}}} - \frac{\vec{p}}{\mu} v}}, \vec{p'} \right)}{\abs{\frac{\vec{q}}{\abs{\vec{q}}} - \frac{\vec{p}}{\mu} v}^{\alpha(\tilde{n},k-1)}} \dd v \ddd{3} p'.
\end{equation}
We now want to show that $2(n - \tilde{n}) + 1 + \alpha(\tilde{n},k-1)$ does not depend on $\tilde{n}$. For $k = 2$ this is straightforward as we see from Eq.~\eqref{eq:scattering-rhon1Absq} that $\alpha(\tilde{n},1) = 2\tilde{n}+1$ thus we get:
\begin{equation}
2(n - \tilde{n}) + 1 + \alpha(\tilde{n},1) = 2n+2.
\end{equation}
Generalizing this assume that $\alpha(\tilde{n}, k-1) = 2n + k$ for some $k \in \NN$, then
\begin{equation}
2(n - \tilde{n}) + 1 + \alpha(\tilde{n},k-1) = 2n + k.
\end{equation}
Since all of our assumptions hold for $k=1$, we get using induction that $\alpha(n,k) = 2n + k$ and
\begin{equation} \label{eq:scattering-rhonkFinal}
\rho_{n,k}(\vec{q},\vec{p}) = \dfrac{\tilde{\rho}_{n,k}(\frac{\vec{q}}{\abs{\vec{q}}}, \vec{p})}{\abs{\vec{q}}^{2n + k}}
\end{equation}
where
\begin{equation}
\tilde{\rho}_{n,k}(\frac{\vec{q}}{\abs{\vec{q}}}, \vec{p}) = \sum_{\tilde{n}=0}^n a_{(n-\tilde{n})} \int\limits_{\RR^3} \int\limits_{0}^{\infty} \tilde{f}_{2(n-\tilde{n})+1}(\frac{\vec{q}}{\abs{\vec{q}}} - \frac{\vec{p}}{\mu} v, \vec{p}, \vec{p'}) \dfrac{\tilde{\rho}_{\tilde{n},k-1} \left( \frac{ \frac{\vec{q}}{\abs{\vec{q}}} - \frac{\vec{p}}{\mu} v}{\abs{\frac{\vec{q}}{\abs{\vec{q}}} - \frac{\vec{p}}{\mu} v}}, \vec{p'} \right)}{\abs{\frac{\vec{q}}{\abs{\vec{q}}} - \frac{\vec{p}}{\mu} v}^{\alpha(\tilde{n},k-1)}} \dd v \ddd{3} p'.
\end{equation}
In order to compute the differential cross section we are only interested in the spatial density of the particles $D(\vec{q}) = \int_{\RR^3} \rho(\vec{q}, \vec{p}) \ddd{3} p$, thus we want to compute
\begin{equation}
D_{n,k}(\vec{q}) = \int_{\RR^3} \rho_{n,k}(\vec{q}, \vec{p}) \ddd{3} p.
\end{equation}
Using Eq.~\eqref{eq:scattering-rhonkFinal} we get
\begin{equation}
D_{n,k}(\vec{q}) = \dfrac{f_{n,k}(\vartheta)}{\abs{\vec{q}}^{2n+k}}
\end{equation}
where $f_{n,k}(\vartheta)$ is some function of the scattering angle $\vartheta$. Due to the symmetry of the scattering problem the spatial density $D_{n,k}(\vec{q})$ cannot depend on the polar angle $\varphi$, but only on the azimuthal angle $\vartheta$, which coincides with the scattering angle. That is why $f_{n,k}(\vartheta)$ depends only on the scattering angle. The density of particles that reach the detector is given as $\lim_{\abs{\vec{q}} \to \infty} D(\vec{q}) \abs{\vec{q}}^2 \dd \Omega$, where $\abs{\vec{q}}^2 \dd \Omega$ is the infinitesimal surface element, $\dd \Omega = \sin(\vartheta) \dd \vartheta \dd \varphi$. We get for $\lambda=1$
\begin{equation}
\lim_{\abs{\vec{q}} \to \infty} D(\vec{q}) \ddd{3} q = \lim_{\abs{\vec{q}} \to \infty} \sum_{n=0}^\infty \sum_{k=0}^\infty \hbar^{2n} D_{n,k}(\vec{q}) \abs{\vec{q}}^2 \dd \Omega
= \lim_{\abs{\vec{q}} \to \infty} \sum_{n=0}^\infty \sum_{k=0}^\infty \hbar^{2n} \dfrac{f_{n,k}(\vartheta)}{\abs{\vec{q}}^{2n+k-2}} \dd \Omega.
\end{equation}
In the limit we get nonzero contributions only from the terms where $2n+k-2 \geq 0$, these are: $n = 0$, $k \leq 2$ and $n = 1$, $k = 0$. Since $\rho_{1,0} = 0$ we have nonzero contribution only from the terms where $n = 0$. But these are only the terms that are obtained using the Poisson bracket, hence $\lim_{\abs{\vec{q}} \to \infty} D(\vec{q}) \abs{\vec{q}}^2 \dd \Omega$ must coincide with the predictions of classical theory.


\end{document}